\documentclass[10pt,reqno]{amsart}
\usepackage{latexsym,amsmath,amssymb,amscd}
\def\today{\ifcase \month \or
   January \or February \or March \or April \or
   May \or June \or July \or August \or
   September \or October \or November \or December \fi
   \space\number\day , \number\year}
   
\makeatletter
  \newcommand\@dotsep{4.5}
  \def\@tocline#1#2#3#4#5#6#7{\relax
     \ifnum #1>\c@tocdepth % then omit
     \else
     \par \addpenalty\@secpenalty\addvspace{#2}%
     \begingroup \hyphenpenalty\@M
     \@ifempty{#4}{%
     \@tempdima\csname r@tocindent\number#1\endcsname\relax
        }{%
         \@tempdima#4\relax
           }%
      \parindent\z@ \leftskip#3\relax \advance\leftskip\@tempdima\relax
      \rightskip\@pnumwidth plus1em \parfillskip-\@pnumwidth
       #5\leavevmode\hskip-\@tempdima #6\relax
       \leaders\hbox{$\m@th
       \mkern \@dotsep mu\hbox{.}\mkern \@dotsep mu$}\hfill
       \hbox to\@pnumwidth{\@tocpagenum{#7}}\par
       \nobreak
        \endgroup
         \fi}
\makeatother 
\begin{document}

%\DeclareRobustCommand{\SkipTocEntry}[4]{} 

\makeatletter
\@addtoreset{figure}{section}
\def\thefigure{\thesection.\@arabic\c@figure}
\def\fps@figure{h,t}
\@addtoreset{table}{bsection}

\def\thetable{\thesection.\@arabic\c@table}
\def\fps@table{h, t}
\@addtoreset{equation}{section}
\def\theequation{%\thesection.
\arabic{equation}}
\makeatother

\newcommand{\bfi}{\bfseries\itshape}

\newtheorem{theorem}{Theorem}
\newtheorem{acknowledgment}[theorem]{Acknowledgment}
\newtheorem{corollary}[theorem]{Corollary}
\newtheorem{definition}[theorem]{Definition}
\newtheorem{example}[theorem]{Example}
\newtheorem{lemma}[theorem]{Lemma}
\newtheorem{notation}[theorem]{Notation}
\newtheorem{proposition}[theorem]{Proposition}
\newtheorem{remark}[theorem]{Remark}

\numberwithin{theorem}{section}
\numberwithin{equation}{section}

\newcommand{\1}{{\bf 1}}
\newcommand{\Ad}{{\rm Ad}}
\newcommand{\Alg}{{\rm Alg}\,}
\newcommand{\Aut}{{\rm Aut}\,}
\newcommand{\ad}{{\rm ad}}
\newcommand{\Borel}{{\rm Borel}}
\newcommand{\Ci}{{\mathcal C}^\infty}
\newcommand{\Cpol}{{\mathcal C}^\infty_{\rm pol}}
\newcommand{\Der}{{\rm Der}\,}
\newcommand{\de}{{\rm d}}
\newcommand{\ee}{{\rm e}}
\newcommand{\End}{{\rm End}\,}
\newcommand{\ev}{{\rm ev}}
\newcommand{\id}{{\rm id}}
\newcommand{\ie}{{\rm i}}
\newcommand{\GL}{{\rm GL}}
\newcommand{\gl}{{{\mathfrak g}{\mathfrak l}}}
\newcommand{\Hom}{{\rm Hom}\,}
\newcommand{\Img}{{\rm Im}\,}
\newcommand{\Ind}{{\rm Ind}}
\newcommand{\Ker}{{\rm Ker}\,}
\newcommand{\Lie}{\text{\bf L}}
\newcommand{\m}{\text{\bf m}}
\newcommand{\pr}{{\rm pr}}
\newcommand{\Ran}{{\rm Ran}\,}
\renewcommand{\Re}{{\rm Re}\,}
\newcommand{\so}{\text{so}}
\newcommand{\spa}{{\rm span}\,}
\newcommand{\Tr}{{\rm Tr}\,}
\newcommand{\Op}{{\rm Op}}
\newcommand{\U}{{\rm U}}

\newcommand{\CC}{{\mathbb C}}
\newcommand{\RR}{{\mathbb R}}
\newcommand{\TT}{{\mathbb T}}

\newcommand{\Ac}{{\mathcal A}}
\newcommand{\Bc}{{\mathcal B}}
\newcommand{\Cc}{{\mathcal C}}
\newcommand{\Dc}{{\mathcal D}}
\newcommand{\Ec}{{\mathcal E}}
\newcommand{\Fc}{{\mathcal F}}
\newcommand{\Hc}{{\mathcal H}}
\newcommand{\Jc}{{\mathcal J}}
\renewcommand{\Mc}{{\mathcal M}}
\newcommand{\Nc}{{\mathcal N}}
\newcommand{\Oc}{{\mathcal O}}
\newcommand{\Pc}{{\mathcal P}}
\newcommand{\Sc}{{\mathcal S}}
\newcommand{\Tc}{{\mathcal T}}
\newcommand{\Vc}{{\mathcal V}}
\newcommand{\Uc}{{\mathcal U}}
\newcommand{\Yc}{{\mathcal Y}}

\newcommand{\Bg}{{\mathfrak B}}
\newcommand{\Fg}{{\mathfrak F}}
\newcommand{\Gg}{{\mathfrak G}}
\newcommand{\Ig}{{\mathfrak I}}
\newcommand{\Jg}{{\mathfrak J}}
\newcommand{\Lg}{{\mathfrak L}}
\newcommand{\Pg}{{\mathfrak P}}
\newcommand{\Sg}{{\mathfrak S}}
\newcommand{\Xg}{{\mathfrak X}}
\newcommand{\Yg}{{\mathfrak Y}}
\newcommand{\Zg}{{\mathfrak Z}}

\newcommand{\ag}{{\mathfrak a}}
\newcommand{\bg}{{\mathfrak b}}
\newcommand{\dg}{{\mathfrak d}}
\renewcommand{\gg}{{\mathfrak g}}
\newcommand{\hg}{{\mathfrak h}}
\newcommand{\kg}{{\mathfrak k}}
\newcommand{\mg}{{\mathfrak m}}
\newcommand{\n}{{\mathfrak n}}
\newcommand{\og}{{\mathfrak o}}
\newcommand{\pg}{{\mathfrak p}}
\newcommand{\sg}{{\mathfrak s}}
\newcommand{\tg}{{\mathfrak t}}
\newcommand{\ug}{{\mathfrak u}}

\newcommand{\ZZ}{\mathbb Z}
\newcommand{\NN}{\mathbb N}
\newcommand{\BB}{\mathbb B}

\newcommand{\ep}{\varepsilon}

\newcommand{\hake}[1]{\langle #1 \rangle }

\newcommand{\scalar}[2]{\langle #1 ,#2 \rangle }
\newcommand{\vect}[2]{(#1_1 ,\ldots ,#1_{#2})}
\newcommand{\norm}[1]{\Vert #1 \Vert }
\newcommand{\normrum}[2]{{\norm {#1}}_{#2}}

\newcommand{\upp}[1]{^{(#1)}}
\newcommand{\p}{\partial}

\newcommand{\opn}{\operatorname}
\newcommand{\slim}{\operatornamewithlimits{s-lim\,}}
\newcommand{\sgn}{\operatorname{sgn}}

\newcommand{\seq}[2]{#1_1 ,\dots ,#1_{#2} }
\newcommand{\loc}{_{\opn{loc}}}

\makeatletter
\title[Magnetic pseudo-differential Weyl calculus on nilpotent Lie groups]{Magnetic pseudo-differential Weyl calculus 
\\ on nilpotent Lie groups}
\author{Ingrid Belti\c t\u a and Daniel Belti\c t\u a}
\address{Institute of Mathematics ``Simion Stoilow'' 
of the Romanian Academy, 
P.O. Box 1-764, Bucharest, Romania}
\email{Ingrid.Beltita@imar.ro}
\email{Daniel.Beltita@imar.ro}
\keywords{pseudo-differential Weyl calculus; magnetic field; Lie group; semidirect product}
\subjclass[2000]{Primary 47G30; Secondary 22E25,22E65,35S05}
%\translator{}
%\dedicatory{}
%\thanks{\textit{File name}: \texttt{top\_Feb1\_2009.tex}}
\date{February 1, 2009}%{\today}
\makeatother

\begin{abstract} 
We develop a pseudo-differential Weyl calculus on nilpotent Lie groups 
which allows one to deal with magnetic perturbations of 
right invariant vector fields. 
For this purpose we investigate 
an infinite-dimensional Lie group constructed as the semidirect product 
of a nilpotent Lie grup and an appropriate function space thereon. 
We single out an appropriate coadjoint orbit in the semidirect product 
and construct our pseudo-differential calculus as a Weyl quantization of that orbit. 
\end{abstract}

\maketitle

\tableofcontents

\section{Introduction}

The Weyl calculus of pseudo-differential operators on $\RR^n$ initiated in \cite{Hor79} is 
a central topic in the theory of linear partial differential equations 
and has been much studied and extended in several directions, 
among which we mention the pseudo-differential Weyl calculus on nilpotent Lie groups systematically developed in~\cite{Me83}. 
In the present paper we focus on a circle of ideas with a similar flavor and 
show that the coadjoint orbits of certain locally convex infinite-dimensional Lie groups 
(in the sense of \cite{Ne06}) can be employed in order to fill in the gap between 
two different lines of investigation motivated by 
the quantum theory: 
\begin{itemize}
\item[-] the magnetic pseudo-differential Weyl calculus on $\RR^n$, initiated independently 
in \cite{KO04} and in \cite{MP04}, and further developed in \cite{IMP07} and other works;  
\item[-] the program of Weyl quantization for coadjoint orbits of some finite-dimensional Lie groups 
including the nilpotent ones (\cite{Wi89}, \cite{Pe94}) and semidirect products 
involving certain semisimple Lie groups
(see \cite{Ca97}, \cite{Ca01}, \cite{Ca07}, and the references therein).
\end{itemize}
Recall that a \emph{magnetic potential} on a Lie group~$G$ is simply a 1-form $A\in\Omega^1(G)$, 
and the corresponding \emph{magnetic field} is $B=dA\in\Omega^2(G)$. 
The purpose of a magnetic pseudo-differential calculus on $G$ is to facilitate the investigation 
on first-order linear differential operators of the form 
\begin{equation}\label{magn_diff_op_G}
-\ie P_0+A(Q)P_0, 
\end{equation}
where 
$P_0$ is a right invariant vector field on $G$ and $A(Q)P_0$ stands 
for the operator defined by the multiplication by the function 
obtained by applying the 1-form $A$ to the vector field $P_0$ at every point in~$G$. 

In the special case of the abelian Lie group $G=(\RR^n,+)$ 
we have $A=A_1\de x_1+\dots+A_n\de x_n\in\Omega^1(\RR^n)$ 
and the operators~\eqref{magn_diff_op_G} on $\RR^n$ 
are precisely the linear partial differential operators determined by the vectors $P_0=(p_1,\dots,p_n)\in\RR^n$, 
\begin{equation}\label{magn_diff_op_RR}
\ie\Bigl(p_1\frac{\partial}{\partial x_1}+\cdots+ p_n\frac{\partial}{\partial x_n}\Bigr)
+\Bigl(p_1A_1(Q)+\cdots+p_nA_n(Q)\Bigr)
=\sum_{j=1}^np_j\Bigl(\ie \frac{\partial}{\partial x_j}+A_j(Q)\Bigr)
\end{equation}
where we denote by $A_1(Q),\dots,A_n(Q)$ the operators of multiplication 
by the coefficients of the 1-form~$A$. 
We refer to \cite{IMP07} for the pseudo-differential calculus 
of the operators~\eqref{magn_diff_op_RR} 
extending the Weyl calculus constructed  
in the non-magnetic case (that is, $A=0$) in the paper~\cite{Hor79}. 

On the other hand, a version of the Weyl calculus for 
right invariant differential operators on nilpotent Lie groups 
has been developed in a series of papers including 
\cite{Me83}, \cite{Mi82}, \cite{Mi86}, \cite{Ma91}, \cite{Pe94}, \cite{Gl04}, \cite{Gl07}, 
and there are remarkable applications of this calculus to various problems 
on partial differential equations on Lie groups. 
See also \cite{An72}, \cite{How77}, \cite{Me81a}, \cite{Me81b}, \cite{HN85} and \cite{BL06} 
for other interesting results related to this circle of ideas. 

For these reasons it is quite natural to try to provide a unifying approach 
to the areas of research mentioned in the preceding two paragraphs. 
It is one of the purposes of the present paper to do that by 
proposing a pseudo-differential calculus on simply connected nilpotent Lie groups 
which takes into account a given magnetic field. 
Our strategy is to pick an appropriate left-invariant space $\Fc$ of functions containing 
the ``coefficients of the magnetic field'' on the Lie group $G$ under consideration 
and then to work within the semidirect product $M=\Fc\rtimes_\lambda G$.  
The latter is in general an infinite-dimensional Lie group,  
and yet we can single out a suitable coadjoint orbit $\Oc$ of $M$ 
which is a \emph{finite-dimensional} symplectic manifold endowed with 
the Kirillov-Kostant-Souriau 2-form 
and is actually symplectomorphic to the cotangent bundle $T^*G$ 
(see Proposition~\ref{orbit1}). 
The spaces of symbols for our pseudo-differential calculus 
will be function spaces on the orbit~$\Oc$, which does not depend 
on the magnetic field. 
However we have to take into account  
a \emph{magnetic predual} $\Oc_*$ for the orbit $\Oc$ 
(Definition~\ref{predual}). 
The set $\Oc_*$ is just a ``copy'' of $\Oc$ contained in the Lie algebra~$\mg$ of 
the infinite-dimensional Lie group $M$ and is the image of $\Oc$ 
by a certain mapping $\theta$ 
defined in terms of a magnetic potential $A\in\Omega^1(G)$. 
In the case $G=(\RR^n,+)$, the mapping $\theta$ is $(x,\xi)\mapsto(\xi+ A(x),x)$. 

In the general case, if two magnetic potentials give rise to the same magnetic field, 
then the corresponding copies of $\Oc$ in the Lie algebra $\mg$ 
are moved to each other by the adjoint action of the Lie group~$M$. 
This leads to the \emph{gauge covariance} of the pseudo-differential calculus 
which we are going to attach to the copy~$\Oc_*$ by the formula 
$$\Op^A(a)f=\int\limits_{\Oc_*}\check{a}(v)\pi(\exp_M v)f\,\de\mu(v)$$
for suitable symbols $a\colon\Oc\to\CC$ and functions $f\colon G\to\CC$. 
(It will be actually convenient  to work with the above integral 
after the change of variables $v=\theta(x,\xi)$ with $(x,\xi)\in T^*G$; 
compare \eqref{weyl_loc_special}~and~\eqref{weyl_loc}.)
Here $\mu$ is the Liouville measure corresponding to the symplectic structure 
on the magnetic predual $\Oc_*\subseteq\mg$ and $\pi$ is a natural irreducible unitary representation 
of the infinite-dimensional Lie group~$M$ on $L^2(G)$ which 
corresponds to the coadjoint orbit~$\Oc$ as in the orbit method (\cite{Ki62}, \cite{Ki76}). 
We show in Theorem~\ref{main1} that the magnetic pseudo-differential Weyl calculus on a nilpotent Lie group $G$ 
possesses appropriate versions of the basic properties 
pointed out in the abelian case $G=(\RR^n,+)$ in \cite{MP04}, 
however the proofs in the present situation 
are considerably more difficult and require 
proving properties of the nilpotent Lie algebras which 
may also have an independent interest (see for instance Proposition~\ref{nilp2}). 
We mention that when $G=(\RR^n,+)$, if $\Fc$ is the $(n+1)$-dimensional vector space of affine functions then one recovers 
the classical Weyl calculus for pseudo-differential operators, 
while for $\Fc=\Cpol(\RR^n)$ the magnetic Weyl calculus of \cite{MP04} is recovered. 

It is noteworthy that, just as in the abelian case, there exists a magnetic Moyal product $\#^A$ 
on the Schwartz space $\Sc(\Oc)$, and ---as a consequence of the gauge covariance--- 
the isomorphism class 
of the associative Fr\'echet algebra $(\Sc(\Oc),\#^A)$ depends only 
on the magnetic field $B=dA\in\Omega^2(G)$.  
Our Theorem~\ref{2step_th} records an explicit formula for $\#^A$ in the case when $G$ is a two-step nilpotent Lie group, 
which extends the corresponding formula established in \cite{MP04} and \cite{KO04} 
and already covers the important situation of the Heisenberg groups. 
We postpone to forthcoming papers both 
the formula for magnetic Moyal product in the case of a general (simply connected) nilpotent Lie group 
and the description and applications of more general classes of symbols 
for the magnetic pseudo-differential Weyl calculus. 
We aim to apply these techniques to more general function spaces~$\Fc$ 
in order to obtain more general radiation conditions for various 
Hamiltonian operators appearing in mathematical physics 
(see for instance \cite{Be01a} and \cite{Be01b}). 

\subsection*{Notation}
Throughout the paper we denote by $\Sc(\Vc)$ the Schwartz space 
on a finite-dimensional real vector space~$\Vc$. 
That is, $\Sc(\Vc)$ is the set of all smooth functions 
that decay faster than any polynomial together with 
their partial derivatives of arbitrary order. 
Its topological dual ---the space of tempered distributions on $\Vc$--- 
is denoted by $\Sc'(\Vc)$. 
We use the notation $\Cpol(\Vc)$ for the space 
of smooth functions that grow polynomially together with 
their partial derivatives of arbitrary order. 
We use $\scalar{\cdot}{\cdot}$ to denote any duality pairing between 
finite-dimensional real vector space whose meaning is clear 
from the context. 
In particular, this may stand for the self-duality given 
a symplectic bilinear form. 

\section{Semidirect products}\label{sect1}

\subsection{One-parameter subgroups in topological groups}

\begin{definition}\label{top1}
\normalfont
For an arbitrary topological group $G$ we define 
$$\Lie(G)=\{X\colon{\mathbb R}\to G\mid X\text{ homomorphism of topological groups}\} $$
and endow this set with the topology of uniform convergence on compact intervals in ${\mathbb R}$. 
The \emph{adjoint action} of $G$ on $\Lie(G)$ is the continuous mapping
$$\Ad\colon G\times\Lie(G)\to\Lie(G),\quad (g,X)\mapsto \Ad(g)X:=gX(\cdot)g^{-1}. $$
The \emph{exponential function} of $G$ is the continuous mapping 
$$\exp_G\colon\Lie(G)\to G,\quad X\mapsto\exp_G X:=X(1).$$
If $H$ is another topological group, then every homomorphism of topological groups $\psi\colon G\to H$ 
induces a continuous mapping 
$\Lie(\psi)\colon\Lie(G)\to\Lie(H),\quad X\mapsto \psi\circ X$ 
and it is easy to see that the diagram 
\begin{equation}\label{nat}
\begin{CD}
\Lie(G) @>{\Lie(\psi)}>> \Lie(H) \\
@V{\exp_G}VV @VV{\exp_H}V \\
G @>{\psi}>> H
\end{CD}
\end{equation}
is commutative. 
In fact $\Lie(\cdot)$ is a functor from the category of topological groups 
to the category of topological spaces, 
and $\exp$ is a natural transformation. 
We refer to \cite{HM05} and Chapter~II in \cite{HM07} 
for these concepts and related results. 
\qed
\end{definition}

\begin{remark}\label{top2}
\normalfont
If $G$ is a finite-dimensional Lie group, then every one-parameter group $X\in\Lie(G)$ 
is actually smooth and there exists a bijective map 
$$\Lie(G)\simeq T_{\1}G$$
which takes every one-parameter subgroup $X\in\Lie(G)$ into 
its infinitesimal generator $\dot{X}(0)\in T_{\1}G$. 
More generally, this assertion holds if $G$ is a locally exponential Lie group 
(modeled on a locally convex space); 
see Def.~II.5.1, Def.~IV.1.1, and Th.~IV.1.18 in \cite{Ne06}. 
\qed
\end{remark}

\begin{remark}\label{top3}
\normalfont
Let $G$ be a topological group and $\Yc$ a complex Banach space.  
We denote 
$$\Cc(\Yc)=\{T\colon\Dc(T)\subseteq\Yc\to\Yc\mid T\text{ closed, densely defined, linear operator}\}.$$
If $\pi\colon G\to\Bc(\Yc)$ is a \so-continuous representation 
which is uniformly bounded 
(that is, $\sup\limits_{g\in G}\Vert\pi(g)\Vert<\infty$), then for every $X\in\Lie(G)$ 
we get a bounded, \so-continuous one-parameter group $\pi\circ X\colon{\mathbb R}\to\Bc(\Yc)$. 
Thus we can define a mapping 
$$\Lie(\pi)\colon\Lie(G)\to\Cc(\Yc),\quad 
X\mapsto\frac{\de}{\de t}\Big\vert_{t=0}\pi(X(t)) $$
by means of the Hille-Yosida theorem, and we have 
\begin{equation}\label{nat_var}
(\forall X\in\Lie(G))\quad  \pi(\exp_G X)=\exp(\Lie(\pi)X)
\end{equation}
(which should be compared with \eqref{nat}). 
Now assume that 
$\Vc$ is a linear subspace of $\Yc$ 
and for every $X\in\Lie(G)$ we have $\Vc\subseteq\Dc(\Lie(\pi)X)$ and 
$(\Lie(\pi)X)\Vc\subseteq\Vc$. 
If moreover $G$ is a topological group with Lie algebra 
in the sense of Chapter~II in \cite{HM07}, 
then it follows by the Trotter formulas that $\Lie(\pi)$ 
induces a representation of the Lie algebra $\Lie(G)$ 
by linear maps on~$\Vc$. 
\qed
\end{remark}

\subsection{Semidirect products and their exponential maps}

\begin{definition}\label{semidirect_def}
\normalfont
Let $G$ be a topological group and $\Fc$ a real topological  vector space 
with the unital associative algebra of continuous endomorphisms denoted by $\End(\Fc)$. 
Assume that $\alpha\colon G\to\End(\Fc)$, $g\mapsto\alpha_g$, 
 is a continuous representation of $G$ on $\Fc$, that is, 
 $\alpha_{\1}=\id_{\Fc}$, $\alpha_{g_1g_2}=\alpha_{g_1}\alpha_{g_2}$ for all $g_1,g_2\in G$, 
 and the mapping 
 $G\times\Fc\to\Fc$, $(g,\phi)\mapsto\alpha_g\phi$
 is continuous. 
 Then the \emph{semidirect product of groups} denoted 
 $\Fc\rtimes_{\alpha}G$ (or $G\ltimes_{\alpha}\Fc$) 
 is the topological group whose underlying topological space is 
 $\Fc\times G$ (respectively $\Fc\rtimes G$) 
 with the multiplication 
\begin{equation}\label{mult}
 (\phi_1,g_1)(\phi_2,g_2)=(\phi_1+\alpha_{g_1}\phi_2,g_1g_2)
\end{equation}
 (respectively $(g_1,\phi_1)(g_2,\phi_2)=(g_1g_2,\phi_1+\alpha_{g_1}\phi_2)$) 
 whenever $g_1,g_2\in G$ and $\phi_1,\phi_2\in\Fc$. 
 It is easy to see that $(0,\1)$ is the unit element in the group $\Fc\times_\alpha G$, 
 while the inversion is given by 
\begin{equation}\label{inv}
 (\phi,g)^{-1}=(-\alpha_{g^{-1}}\phi,g^{-1})
\end{equation}
 for every $\phi\in\Fc$ and $g\in G$. 
 
 Now let $\gg$ be any real topological Lie algebra and assume that 
 $\dot{\alpha}\colon\gg\to\End(\Fc)$, $X\mapsto\dot{\alpha}(X)$,
  is a continuous representation of $\gg$ on $\Fc$, 
 that is, $\dot{\alpha}$ is a linear mapping such that 
 $\dot{\alpha}([X_1,X_2])=[\dot{\alpha}(X_1),\dot{\alpha}(X_2)]
 :=\dot{\alpha}(X_1)\dot{\alpha}(X_2)-\dot{\alpha}(X_2)\dot{\alpha}(X_1)$ 
 for all $X_1,X_2\in\gg$ and the mapping 
 $\gg\times\Fc\to\Fc$, $(X,\phi)\mapsto\dot{\alpha}(X)\phi$ 
 is continuous. 
 Then the \emph{semidirect product of Lie algebras} denoted $\Fc\rtimes_{\dot{\alpha}}\gg$ 
 is the topological Lie algebra whose underlying topological vector space is 
 $\Fc\times\gg$ with the bracket 
 \begin{equation}\label{semidirect_bracket}
 [(\phi_1,X_1),(\phi_2,X_2)]=(\dot{\alpha}(X_1)\phi_2-\dot{\alpha}(X_2)\phi_1,[X_1,X_2])
 \end{equation}
 for every $X_1,X_2\in\gg$ and $\phi_1,\phi_2\in\Fc$. 
 One can similarly define the semidirect product of Lie algebras $\gg\ltimes_{\dot{\alpha}}\Fc$. 
\qed
\end{definition}

\begin{remark}\label{semidirect_rem}
\normalfont
In the setting of Definition~\ref{semidirect_def}, 
if $G$ is a locally convex Lie group (see \cite{Ne06}), 
$\Fc$ is a complete locally convex vector space and 
the mapping $G\times\Fc$, $(g,\phi)\mapsto\alpha_g\phi$ is smooth, 
then it is straightforward to prove the following assertions: 
\begin{enumerate}
\item\label{semidirect_rem_item1}
The semidirect product $M:=\Fc\rtimes_{\alpha}G$ is 
a locally convex Lie group whose Lie algebra is 
$\mg:=\Fc\rtimes_{\dot{\alpha}}\gg$, where $\gg=\Lie(G)$ is the Lie algebra of $G$ 
and $\dot\alpha\colon\gg\to\End(\Fc)$ is defined by the condition that 
for every $\phi\in\Fc$ the linear mapping $\gg\to\Fc$, $X\mapsto\dot{\alpha}(X)\phi$ 
is the differential of the smooth mapping $G\to\Fc$, $g\mapsto\alpha_g\phi$ at 
the point~$\1\in G$. 
\item\label{semidirect_rem_item2}
The adjoint action of the Lie group $M$ 
on its Lie algebra $\mg$ is given by 
$$\Ad_M\colon M\times\mg\to\mg,\quad
(\Ad_M(\phi,g))(\psi,X)=(\alpha_g\psi-\dot{\alpha}(\Ad_G(g)X)\phi,\Ad_G(g)X)$$
for $(\phi,g)\in\Fc\rtimes_{\alpha}G=M$ and $(\psi,X)\in\Fc\rtimes_{\dot{\alpha}}\gg=\mg$. 
\item\label{semidirect_rem_item3}
The coadjoint action of the Lie group $M$ 
on the dual of its Lie algebra $\mg^*=\Fc^*\times\gg^*$ is given by
$$\Ad_M^*\colon M\times\mg^*\to\mg^*,\quad 
(\Ad_M^*(\phi,g))(\nu,\xi)=(\alpha_{g^{-1}}^*\nu,\Ad_G^*(g)\xi+\dot{\alpha}_\phi^*\alpha_{g^{-1}}^*\nu)
$$
for $(\phi,g)\in\Fc\rtimes_{\alpha}G=M$ and  $(\nu,\xi)\in\Fc^*\times\gg^*=\mg^*$, 
where $\alpha_\phi^*\colon\Fc^*\to\gg^*$ is the dual of the linear mapping 
$\alpha_\phi:=\dot{\alpha}(\cdot)\phi\colon\gg\to\Fc$ (see item~\eqref{semidirect_rem_item1} above). 
\end{enumerate}
\qed
\end{remark}

\begin{example}\label{top3.5}
\normalfont
Let $n\ge1$ and assume that $\Fc$ is a linear subspace of 
the space of real Borel functions $\Bc_{\mathbb R}({\mathbb R}^n)$ 
which is invariant under translations and 
is endowed with a linear topology such that 
the mapping 
\begin{equation*}%\label{action}
{\mathbb R}^n\times\Fc\to \Fc,\quad (q,f)\mapsto \alpha(q)f:=f(q+\cdot)
\end{equation*}
is continuous. 
If we denote by $\alpha$ the corresponding action of the additive group $({\mathbb R}^n,+)$ 
by endomorphisms of the group $(\Fc,+)$, then we can construct the semi-direct product 
$$G:=\Fc\rtimes_\alpha{\mathbb R}^n,$$
which is a topological group with the multiplication defined by~\eqref{mult}. 
Moreover, $G$ has a natural unitary representation on the Hilbert space $\Hc:=L^2(\RR^n)$, 
defined by 
\begin{equation}\label{unirep}
\pi\colon G\to\Bc(\Hc),\quad \pi(f,q)\phi=\ee^{\ie f}\phi(q+\cdot)
\text{ whenever }\varphi\in\Hc,\ f\in\Fc,\text{ and }q\in{\mathbb R}^n. 
\end{equation}
If the topology of the function space $\Fc$ is stronger than the topology of pointwise convergence, 
then it follows by Lebesgue's dominated convergence theorem 
that the representation $\pi$ is \so-continuous. 

Here are some special cases of this construction: 
\begin{enumerate}
\item 
For any integer $k\ge1$ let us consider 
the following space of polynomial functions on ${\mathbb R}^n$
$$\Pc_k({\mathbb R}^n)=\{f\in{\mathbb R}[q_1,\dots,q_n]\mid\deg f\le k\}.$$
The linear space $\Pc_k({\mathbb R}^n)$ is finite-dimensional and 
is invariant under translations, hence we can form 
the semi-direct product 
$G_k:=\Pc_k({\mathbb R}^n)\rtimes_{\alpha}{\mathbb R}^n$, 
which is a finite-dimensional, nilpotent, simply connected Lie group. 
The special case $k=1$ of this construction is particularly important, 
since $G_1$ is precisely the $(2n+1)$-dimensional Heisenberg group. 
\item
If $\Fc=\Ci_{\mathbb R}({\mathbb R}^n)$ with the natural Fr\'echet topology, 
then it follows by Ex.~II.5.9 in \cite{Ne06} that 
$G=\Ci_{\mathbb R}({\mathbb R}^n)\rtimes_\alpha{\mathbb R}^n$ is a (Fr\'echet-)Lie group 
whose Lie algebra is the semi-direct product
$$\gg=\Ci_{\mathbb R}({\mathbb R}^n)\rtimes_{\dot{\alpha}}{\mathbb R}^n,$$
where 
$$\dot{\alpha}\colon{\mathbb R}^n\to\Der(\Ci_{\mathbb R}({\mathbb R}^n)),\quad 
(p_1,\dots,p_n)\mapsto p_1\frac{\partial}{\partial q_1}+\cdots
+p_n\frac{\partial}{\partial q_n}.$$
The Lie algebra $\gg$ fails to be abelian or even nilpotent, however it is solvable 
since $[\gg,\gg]=\Ci_{\mathbb R}({\mathbb R}^n)$, hence $[[\gg,\gg],[\gg,\gg]]=\{0\}$. 
As regards the finite-dimensional Lie groups 
$G_k:=\Pc_k({\mathbb R}^n)\rtimes_{\alpha}{\mathbb R}^n$ for $k\ge1$, 
we also note that 
$G_1\subset G_2\subset\cdots\subset\overline{\bigcup\limits_{k\ge1}G_k}=G$.
\end{enumerate}
\qed
\end{example}

The following statement partially extends 
Th.~49.6 and remark~38.9 in \cite{KM97} and 
some facts noted in Ex.~II.5.9 in \cite{Ne06}. 
See also Sect.~3 in \cite{MS03} for the expression of the exponential map 
for a semi-direct product of finite-dimensional Lie groups. 

\begin{proposition}\label{top3.6}
Let $G$ be a topological group acting on a topological space $D$ 
by an action denoted simply by 
$$G\times D\to D, \quad (g,x)\mapsto g.x$$ 
and assume that $\Fc$ is a linear subspace of 
the space of real Borel functions $\Bc_{\mathbb R}(D)$ 
which is invariant under the translation operators 
$\alpha_g\colon\Bc_{\mathbb R}(D)\to\Bc_{\mathbb R}(D)$ defined by  
$(\alpha_g\phi)(x)=\phi(g^{-1}.x)$ for $g\in G$, $x\in D$, and 
$\phi\in\Bc_{\mathbb R}(D)$. 
Also assume that $\Fc$ 
is endowed with a complete, locally convex topology such that 
the mapping 
\begin{equation}\label{action}
G\times\Fc\to \Fc,\quad (g,\phi)\mapsto \alpha_g\phi
\end{equation}
is continuous. 
Then the following assertions hold: 
\begin{enumerate}
\item\label{top3.6_item1}
The mapping 
\begin{equation}\label{beta}
\Fc\times\Lie(G)\to \Fc,\quad 
(\phi,X)\mapsto \beta(X)\phi:=\int\limits_0^1\alpha_{X(s)}\phi\,\de s
\end{equation}
is well defined and continuous. 
\item\label{top3.6_item2} 
For every pair $(\phi,X)\in\Fc\times\Lie(G)$, the function 
$$Z_{\phi,X}\colon{\mathbb R}\to\Fc\times G,\quad 
Z_{\phi,X}(t)=(t\beta(tX)\phi,X(t))
=\Bigl(\int\limits_0^t\alpha_{X(s)}\phi\,\de s,X(t)\Bigr)  $$
has the property $Z_{\phi,X}\in\Lie(\Fc\rtimes_\alpha G)$. 
Moreover, $t\mapsto t\beta(tX)\phi$ is a differentiable curve in $\Fc$ 
and $\frac{\de}{\de t}\Big\vert_{t=0}(t\beta(tX)\phi)=\phi$. 
\item\label{top3.6_item2.5}
Let $\psi\in\Fc$ and $X\in\Lie(G)$ such that 
the curve ${\mathbb R}\to\Fc$, $t\mapsto\alpha_{X(t)}\psi$ is differentiable, 
and denote 
$\dot{\alpha}(X)\psi:=\frac{\de}{\de t}\Big\vert_{t=0}\alpha_{X(t)}\psi\in\Fc$. 
Then 
$$(\forall \phi\in\Fc)\quad 
\bigl(\Ad_{\Fc\rtimes_\alpha G}\psi\bigr)Z_{\phi,X}
=Z_{\phi-\dot{\alpha}(X)\psi,X}\in\Lie(\Fc\rtimes_{\alpha}G).$$
\item\label{top3.6_item3}
If we assume that $G$ is a finite-dimensional Lie group acting on itself 
by left translations (hence $D=G$ and $(\alpha_g\phi)(x)=(\lambda_g\phi)(x)=\lambda(g^{-1}x)$ 
for $g,x\in G$ and $\phi\in\Fc$) and 
there exists the continuous inclusion 
$\Fc\hookrightarrow\Ci(G)$ such that 
the mapping \eqref{action} is smooth, then 
$\Fc\rtimes_\lambda G$ is a locally convex Lie group 
with the following properties: 
\begin{enumerate}
\item\label{top3.6_item3_item-a}
The Lie algebra of $\Fc\rtimes_\lambda G$ is 
the semi-direct product of Lie algebras 
$\Fc\rtimes_{\dot{\lambda}}\gg$, 
where $\gg:=\Lie(G)$, $\Fc$ is thought of as an abelian Lie algebra 
and the mapping 
$\dot{\lambda}\colon\gg\to\Der(\Fc)$ is defined 
as in~\eqref{top3.6_item2.5} above. 
(That is, $\dot{\lambda}$ is induced by the natural representation 
of the elements in $\gg$ as right-invariant vector fields on $G$.)
\item\label{top3.6_item3_item-b} 
The exponential map of the Lie group $\Fc\rtimes_\lambda G$ 
is defined by the formula 
$$\exp_{\Fc\rtimes_\lambda G}\colon
\Fc\rtimes_{\dot{\lambda}}\gg\to 
\Fc\rtimes_\lambda G,\quad 
(\phi,X)\mapsto(\beta(X)\phi,\exp_G X).$$
\item\label{top3.6_item3_item-c}
Assume $G=({\mathbb R}^n,+)$ with the generic point denoted by 
$(q_1,\dots,q_n)$.  If 
$A_j,A_j,\psi\in\Fc$ and $j\in\{1,\dots,n\}$ satisfy 
$A_j'=A_j+{\partial\psi}/{\partial q_j}$, 
then 
$$\bigl(\Ad(\exp_{\Fc\rtimes_\lambda{\mathbb R}^n}\psi)\bigr)(A'_j,p_j)
=(A_j,p_j)
\in\Fc\rtimes_{\dot{\lambda}}{\mathbb R}^n.$$
\end{enumerate}
\end{enumerate}
\end{proposition}

\begin{proof}
\eqref{top3.6_item1}
For every $(\phi,X)\in\Fc\times G$ the function 
$[0,1]\to\Fc$, $s\mapsto\alpha_{X(s)}\phi$ 
is Riemann integrable since it is continuous and 
the locally convex space $\Fc$ is complete; 
see for instance Lemma~2.5 in Ch.~I of \cite{KM97}. 
The continuity of the mapping $(\phi,X)\mapsto\beta(X)\phi$ follows 
by the continuity of~\eqref{action} and 
the continuity properties of the Riemann integral. 

\eqref{top3.6_item2}
The second equality in the definition of $Z_{\phi,X}(t)$ follows by 
a change of variables in the Riemann integral 
(Corollary~2.6(3) in Ch.~I of \cite{KM97}), 
and $Z_{\phi,X}\colon{\mathbb R}\to\Fc\times G$ is continuous 
by the previous assertion~\eqref{top3.6_item1}. 
Moreover, for arbitrary $t_1,t_2\in{\mathbb R}$ we have 
$$\begin{aligned}
Z_{\phi,X}(t_1)Z_{\phi,X}(t_2)
&=(t_1\beta(t_1X)\phi,X(t_1))(t_2\beta(t_2X)\phi,X(t_2)) \\
&=(t_1\beta(t_1p)\phi+\alpha_{X(t_1)}(t_2\beta(t_2X)\phi),X(t_1)X(t_2)) \\
&=(t_1\beta(t_1X)\phi+t_2\alpha_{X(t_1)}\beta(t_2X)\phi,X(t_1+t_2)). 
  \end{aligned}
$$
On the other hand, 
$$\begin{aligned}
t_1\beta(t_1X)\phi+t_2\alpha_{X(t_1)}\beta(t_2X)\phi
&=t_1\int\limits_0^1\alpha_{X(t_1s)}\phi\,\de s
 +t_2\alpha_{X(t_1)}\int\limits_0^1\alpha_{X(t_2s)}\phi\,\de s \\
&=\int\limits_0^{t_1}\alpha_{X(s)}\phi\,\de s
 +\alpha_{X(t_1)}\int\limits_0^{t_2}\alpha_{X(s)}\phi\,\de s \\
&= \int\limits_0^{t_1}\alpha_{X(s)}\phi\,\de s
 +\int\limits_0^{t_2}\alpha_{X(t_1+s)}\phi\,\de s \\
&=\int\limits_0^{t_1}\alpha_{X(s)}\phi\,\de s
 +\int\limits_{t_1}^{t_1+t_2}\alpha_{X(s)}\phi\,\de s \\
&=\int\limits_0^{t_1+t_2}\alpha_{X(s)}\phi\,\de s \\
&=(t_1+t_2)\beta((t_1+t_2)X)\phi
  \end{aligned}
$$
and it follows that $Z_{\phi,X}(t_1)Z_{\phi,X}(t_2)=Z_{\phi,X}(t_1+t_2)$. 
Thus $Z_{\phi,X}\in\Lie(\Fc\rtimes_\alpha G)$.

The equality $\frac{\de}{\de t}\Big\vert_{t=0}(t\beta(tX)\phi)=\phi$ 
follows by 
Lemma~2.5 in Ch.~I of \cite{KM97} again. 

\eqref{top3.6_item2.5} 
Firstly note that 
$\alpha_{X(t_1+t_2)}=\alpha_{X(t_1)}\alpha_{X(t_2)}$ 
for every $t_1,t_2\in{\mathbb R}$, hence we have 
\begin{equation}\label{diff}
(\forall s\in{\mathbb R})\quad 
\frac{\de}{\de t}\Big\vert_{t=s}\alpha_{X(t)}\psi
=\alpha_{X(s)}\dot{\alpha}(X)\psi. 
\end{equation}
On the other hand, it follows by \eqref{mult}~and~\eqref{inv} that 
$(\psi,\1)^{-1}=(-\psi,\1)$ 
and $(\psi,\1)(\phi,g)(\psi,\1)^{-1}=(\psi+\phi-\alpha_g\psi,g)$ 
whenever $\phi\in\Fc$ and $g\in G$. 
Therefore for arbitrary $\phi\in\Fc$ and $t\in{\mathbb R}$ we get 
$$\begin{aligned}
((\Ad_{\Fc\rtimes_\alpha G}\psi)Z_{\phi,X})(t) 
 &= (\psi,\1)Z_{\phi,X}(t)(\psi,\1)^{-1} 
  = (\psi,\1) \Bigl(\int\limits_0^t\alpha_{X(s)}\phi\,\de s,X(t)\Bigr) 
     (\psi,\1)^{-1} \\
 &= \Bigl(\psi+\int\limits_0^t\alpha_{X(s)}\phi\,\de s-\alpha_{X(t)}\psi,
     X(t)\Bigr) \\
 &= \Bigl(\int\limits_0^t\alpha_{X(s)}\phi\,\de s
  -\int\limits_0^t\frac{\de}{\de r}\Big\vert_{r=s}(\alpha_{X(r)}\psi)\,\de s, 
   X(t)\Bigr) \\
 &= \Bigl(\int\limits_0^t\alpha_{X(s)}(\phi-\dot{\alpha}(X)\psi)\,\de s, X(t)\Bigr) \\
 &= Z_{\phi-\dot{\alpha}(X)\psi,X}(t), 
  \end{aligned}
$$
where the next-to-last equality follows by~\eqref{diff}. 

\eqref{top3.6_item3} 
Let us denote $M=\Fc\rtimes_\lambda G$ and $\mg=\Fc\rtimes_{\dot\lambda}\gg$. 
It is clear that $\mg= T_{(0,\1)}S$ so in order to prove that 
$\mg=\Lie(S)$, we still have to check that the operations 
of sum and bracket in these spaces agree. 
The latter fact follows since for every $(\phi,X)\in\mg$ 
and every $t\in{\mathbb R}$ we have 
$\exp_M(t(\phi,X))=Z_{\phi,X}(t)$ by means of the above 
assertion~\eqref{top3.6_item2}. 
This shows that \eqref{top3.6_item3_item-a}-\eqref{top3.6_item3_item-b} 
hold. 
The remaining property \eqref{top3.6_item3_item-c}  follows by assertion~\eqref{top3.6_item2.5}. 
\end{proof}

\subsection{Coadjoint orbits of semidirect products}

The symplectic structures on coadjoint orbits of semidirect products 
defined by \emph{finite-dimensional} representations of Lie groups
were thoroughly investigated in \cite{Ba98}.
As we are interested in semidirect product $M=\Fc\rtimes_\lambda G$, where
$\lambda\colon G\to \End(\Fc)$ is a representation on a function space 
$\Fc$, which is in general infinite dimensional,  in this subsection we shall study a coadjoint orbit $\Oc$ of $M$ that is not covered by the results in the 
of \cite{Ba98}.
This orbit will play a central role in our construction of magnetic pseudo-differential operators.

\begin{definition}\label{orbit0}
\normalfont
Let $G$ be a finite-dimensional Lie group and $\Fc$ a linear subspace of $\Bc_{\RR}(G)$ 
endowed with a locally convex topology. 
We say that the function space $\Fc$ is \emph{admissible} 
if it satisfies the following conditions: 
\begin{enumerate}
\item\label{orbit0_item1}
The linear space $\Fc$ is invariant under the representation of $G$ by left translations, 
$$\lambda\colon G\to\End(\Bc_{\RR}(G)),\quad (\lambda_g\phi)(x)=\phi(g^{-1}x).$$
That is, if $\phi\in\Fc$ and $g\in G$ then $\lambda_g\phi\in\Fc$. 
We denote again by $\lambda\colon G\to\End(\Fc)$ the restriction to $\Fc$ 
of the aforementioned representation of~$G$. 
\item\label{orbit0_item2}
We have $\Fc\subseteq\Ci(G)$ and the topology of $\Fc$ is stronger than 
the topology induced from $\Ci(G)$. 
In other words, the inclusion mapping $\Fc\hookrightarrow\Ci(G)$ is continuous. 
\item\label{orbit0_item3}
The mapping $G\times\Fc\to\Fc$, $(g,\phi)\mapsto\lambda_g\phi$ is smooth. 
For every $\phi\in\Fc$ we denote by $\dot{\lambda}(\cdot)\phi\colon\gg\to\Fc$ 
the differential of the mapping $g\mapsto\lambda_g\phi$ at the point $\1\in G$. 
Thus for all $X\in\gg$ and $g\in G$ we have 
\begin{equation}\label{lambda_diff}
\begin{aligned}
(\dot{\lambda}(X)\phi)(g)
&=\frac{\de}{\de t}\Big\vert_{t=0}\phi(\exp_G(-tX)g)
=-(\phi\circ R_g)'_0(X) \\
&=-(\phi'_g\circ (R_g)'_0)(X)
=-\langle ((R_g)'_0)^*(\phi'_g),X\rangle
\end{aligned}
\end{equation}
where $\langle\cdot,\cdot\rangle\colon\gg^*\times\gg\to\RR$ 
is the canonical duality pairing and $R_g\colon G\to G$, $x\mapsto xg$. 
\item\label{orbit0_item4} 
The points in $G$ are separated by the functions in $\Fc$, that is, 
for every $g_1,g_2\in G$ with $g_1\ne g_2$ there exists $\phi\in\Fc$ with $\phi(g_1)\ne\phi(g_2)$. 
\item\label{orbit0_item5} 
We have $\{\phi'_g\mid\phi\in\Fc\}=T_g^*G$ for every $g\in G$. 
\end{enumerate}
It is clear that $\Ci(G)$ itself is admissible. 
\qed
\end{definition}

\begin{proposition}\label{orbit1}
Let $G$ be a finite-dimensional Lie group and $\Fc\hookrightarrow\Ci(G)$ 
an admissible function space on~$G$. 
Denote $M=\Fc\rtimes_\lambda G$, $\mg=\Lie(M)$ and for every $g\in G$ 
let $\delta_g\colon\Fc\to\RR$, $\phi\mapsto\phi(g)$. 
Define
$$\Oc:=\{(\delta_g,\xi)\mid g\in G,\xi\in\gg^*\}\subseteq\Fc^*\times \gg^*=\mg^*.$$ 
Then $\Oc$ is a coadjoint orbit of the locally convex Lie group $M$ 
which has the following properties: 
\begin{enumerate}
\item\label{orbit1_item1}
The orbit $\Oc$ is a smooth finite-dimensional manifold such that 
for every $\mu\in\Oc$ 
the coadjoint action defines a trivial smooth  bundle 
$\Pi_\mu\colon M\to\Oc$, $m\mapsto\Ad_M^*(m)\mu$. 
\item\label{orbit1_item2}
There exists a canonical symplectic form 
$\omega\in\Omega^2(\Oc)$ invariant under the coadjoint action of $M$ on $\Oc$, 
such that for every $\mu\in\Oc$ the pull-back $\Pi_\mu^*(\omega)\in\Omega^2(M)$ 
is a left invariant 2-form on $M$ whose value at $\1\in M$ is the bilinear functional 
$(\Pi_\mu^*(\omega))_{\1}\colon\mg\times\mg\to\RR$, $(X,Y)\mapsto -\mu([X,Y])$. 
\item\label{orbit1_item3}
The symplectic manifold $(\Oc,\omega)$ is 
symplectomorphic to the cotangent bundle $T^*G$ endowed with its canonical symplectic structure. 
\end{enumerate}
\end{proposition}

\begin{proof}
Denote $\tilde{\delta}_{\1}:=(\delta_{\1},0)\in\Oc$. 
It follows by Remark~\ref{semidirect_rem}\eqref{semidirect_rem_item3} that 
for an arbitrary element $(\phi,g)\in M=\Fc\rtimes_\lambda G$ we have 
\begin{equation}\label{coadj}
\begin{aligned}
\Pi_{\tilde{\delta}_{\1}}(\varphi,g)
=(\Ad_M^*(\phi,g))\tilde{\delta}_{\1}
&=(\lambda_{g^{-1}}^*(\delta_{\1}),\dot{\lambda}_\phi^*(\lambda_{g^{-1}}(\delta_{\1}))) \\
&=(\delta_g,\dot{\lambda}_\phi^*(\delta_g))
=(\delta_g,((R_g)'_0)^*(\phi'_g))\in\Fc^*\times\gg^* 
\end{aligned}
\end{equation}
since, if we denote again by $\langle\cdot,\cdot\rangle\colon\gg^*\times\gg\to\RR$ 
the canonical duality pairing, then 
for every $X\in\gg$ we get 
$$\langle \dot{\lambda}_\phi^*(\delta_g)),X\rangle
=\delta_g(\lambda_\phi(X))
=(\dot{\lambda}(X)\phi)(g)
=-\langle ((R_g)'_0)^*(\phi'_g),X\rangle,$$
where the latter equality follows by~\eqref{lambda_diff}. 
Now note that $((R_g)'_0)^*\colon T_g^*G\to T_{\1}^*G=\gg^*$ is a linear isomorphism, 
hence by \eqref{coadj} and condition \eqref{orbit0_item5} in Definition~\ref{orbit0} 
we get $\{(\Ad_M^*(\phi,g))\tilde{\delta}_{\1}\mid (\phi,g)\in M\}=\Oc$, 
hence the set $\Oc$ is indeed a coadjoint orbit in $\mg^*$. 

We now proceed to proving the other properties of $\Oc$ mentioned in the statement. 
Note that the natural surjective mapping 
\begin{equation}\label{struct}
T^*G\to\Oc,\quad (g,\xi)\mapsto(\delta_g,\xi)
\end{equation}
is also injective since points of $G$ are separated by the functions in $\Fc$ 
(property~\eqref{orbit0_item4} in Definition~\ref{orbit0}). 
We shall endow $\Oc$ with the structure of smooth finite-dimensional manifold 
such that the mapping~\eqref{struct} is a diffeomorphism. 
Let $\omega\in\Omega^2(\Oc)$ be the symplectic form obtained by 
transporting the canonical symplectic form of $T^*G$ by means 
of the diffeomorphism~\eqref{struct}. 

In order to describe $\omega$, 
recall that $T^*G$ is a trivial vector bundle over $G$ with the fiber $\gg^*$ and, by using the left trivialization,  
we may perform the identification $T^*G=G\ltimes_{\Ad_G^*}\gg^*$. 
This makes $T^*G$ into a finite-dimensional Lie group 
whose Lie algebra is $\Lie(T^*G)=\gg\ltimes_{\ad_{\gg}^*}\gg^*$.  
Then the tangent bundle $T(T^*G)=T^*G\ltimes_{\Ad_{T^*G}}\Lie(T^*G)$ 
is a trivial bundle over $T^*G$ with the fiber $\Lie(T^*G)$ using again the left trivialization, 
hence 
$$T(T^*G)=T^*G\times(\gg\times\gg^*)=(G\times\gg^*)\times(\gg\times\gg^*)$$
with the natural projection $T(T^*G)\to T^*G$ given by $((g_0,\xi_0),(X,\xi))\mapsto(g_0,\xi_0)$. 
Then the Liouville 1-form $\sigma\in\Omega^1(T^*G)$ is 
$\sigma\colon T(T^*G)\to\RR$, $((g_0,\xi_0),(X,\xi))\mapsto\langle\xi_0, X\rangle$, 
and 
the canonical symplectic form on $T^*G$ is $-d\sigma\in\Omega^2(T^*G)$ 
(see for instance Ch.~V, \S 7 in \cite{La01}, Ex.~43.9 in \cite{KM97},  or  subsection~6.5 in \cite{CW99}). 
It is easily seen that the value of the 2-form $-d\sigma$ on 
$T_{(g_0,\xi_0)}(T^*G)\simeq\gg\times\gg^*$ is 
given by 
\begin{equation}\label{canonical}
-(d\sigma)_{(g_0,\xi_0)}\colon T_{(g_0,\xi_0)}(T^*G)\times T_{(g_0,\xi_0)}(T^*G)\to\RR,\quad 
((X_1,\xi_1),(X_2,\xi_2))\mapsto\langle\xi_2, X_1\rangle-\langle\xi_1, X_2\rangle .
\end{equation}
Note that the symplectic 2-form $-d\sigma$ is invariant 
under the action of the Lie group $T^*G$ on itself under left translations, 
while the 1-form $\eta$ is not. 
(See also \cite{Li86}.) 

For arbitrary $\mu\in\Oc$ let 
$$M_\mu:=\{m\in M\mid\Ad_M^*(m)\mu=\mu\}$$
be the corresponding coadjoint isotropy group. 
It follows by \eqref{coadj} that 
\begin{equation}\label{isotropy}
M_{\tilde{\delta}_{\1}}
=\{\varphi\in\Fc\mid \phi'_{\1}=0\}\times\{\1\}\subseteq\Fc\rtimes_\lambda G=M.
\end{equation}
We now prove that the smooth mapping 
$\Pi_{\tilde{\delta}_{\1}}\colon M\to\Oc$, $m\mapsto\Ad_M^*(m)\tilde{\delta}_{\1}$ 
 is a trivial bundle with the fiber~$M_{\tilde{\delta}_{\1}}$. 
In fact, since $\dim\gg^*<\infty$, it easily follows by condition~\eqref{orbit0_item5} 
in Definition~\ref{orbit0} that there exists a linear mapping 
$\gg^*\to\Fc$, $\xi\mapsto\varphi_\xi$ such that 
for every $\xi\in\gg^*$ we have $(\varphi_\xi)'_{\1}=\xi$. 
If $\phi,\chi,\psi\in\Fc$, $\psi'_{\1}=0$, and $g\in G$, 
then the equation $(\phi,g)(\psi,\1)=(\chi,g)$ in $M=\Fc\rtimes_\lambda G$ 
is equivalent to $\phi+\lambda_g\psi=\chi$, 
whence $\lambda_{g^{-1}}\phi+\psi=\lambda_{g^{-1}}\chi$. 
Since $\psi'_{\1}=0$, it then follows 
$(\lambda_{g^{-1}}\phi)'_{\1}=(\lambda_{g^{-1}}\chi)'_{\1}$. 
This equation is satisfied for 
$\phi=\lambda_g(\phi_\xi)\in\Fc$, where $\xi:=(\lambda_{g^{-1}}\chi)'_{\1}$. 
Then we can take $\psi:=\lambda_{g^{-1}}\chi-\phi_\xi=\lambda_{g^{-1}}\chi-\lambda_{g^{-1}}\phi$. 
This shows that the smooth cross-section of $\Pi_{\tilde{\delta}_{\1}}$ defined by 
$$\Oc\to\Fc\rtimes_\lambda G,\quad 
(\delta_g,\xi)\mapsto(\lambda_g(\phi_\xi),g)$$
has the property that every element in $\Fc\rtimes_\lambda G$ 
can be uniquely factorized as the product of an element in the image of this cross section 
and an element in the isotropy subgroup $M_{\tilde{\delta}_{\1}}$. 
This implies that $\Pi_{\tilde{\delta}_{\1}}\colon M\to\Oc$ is a trivial bundle. 
For an arbitrary element $\mu\in\Oc$ let $m\in M$ such that $\Ad_M^*(m)\tilde{\delta}_{\1}=\mu$. 
Then the inner automorphism $\Psi\colon M\to M$, $n\mapsto mnm^{-1}$ 
has the property $\Psi(M_{\tilde{\delta}_{\1}})=M_\mu$, whence we easily get 
a factorization property in $M$ with respect to $M_\mu$, 
similar to the one just proved for $M_{\tilde{\delta}_{\1}}$. 
Thus the smooth mapping $\Pi_\mu\colon M\to\Oc$, $m\mapsto\Ad_M^*(m)\mu$, 
is a trivial bundle with the fiber~$M_\mu$. 
It then follows that the classical Kirillov-Kostant-Souriau construction of symplectic forms 
on coadjoint orbits works 
(see for instance Example~4.31 in \cite{Be06}) 
and leads to a symplectic form $\widetilde{\omega}\in\Omega^2(\Oc)$ 
with the properties mentioned in assertion~\eqref{orbit1_item2} in the statement. 

To complete the proof we still have to show that 
the symplectic forms $\omega,\widetilde{\omega}\in\Omega^2(\Oc)$ constructed so far 
actually coincide. 
It follows by \eqref{coadj} that if we identify $\Oc$ to $T^*G$ 
by means of the mapping \eqref{struct}, 
then the differential of the mapping $\Pi_{\tilde{\delta}_{\1}}$ 
at $(0,\1)\in M$ is the linear map
$$\mg=\Fc\rtimes_{\dot{\lambda}}\gg\to T_{(\1,0)}(T^*G)\simeq\gg\times\gg^*,\quad 
(\phi,X)\mapsto(X,\phi'_0).$$
Then \eqref{canonical} shows that the value of the 2-form 
$\Pi_{\tilde{\delta}_{\1}}^*(\omega)=\Pi_{\tilde{\delta}_{\1}}^*(-d\sigma)$ at $(0,\1)\in M$ 
is the bilinear functional 
$$\mg\times\mg\to\RR,\quad 
((\phi_1,X_1),(\phi_2,X_2))\mapsto \langle(\phi_2)'_0,X_1\rangle-\langle(\phi_1)'_0,X_2\rangle
=-\tilde{\delta}_{\1}([(\phi_1,X_1),(\phi_2,X_2)]) $$
(see~\eqref{semidirect_bracket}). 
Thus $\Pi_{\tilde{\delta}_{\1}}^*(\omega)=\Pi_{\tilde{\delta}_{\1}}^*(\widetilde{\omega})$ 
on $\mg=T_{(0,\1)}M$, 
and then $\omega=\widetilde{\omega}$ on $T_{\tilde{\delta}_{\1}}\Oc\simeq T_{(\1,0)}(T^*G)$. 
By using the fact that $\Pi_{\tilde{\delta}_{\1}}\colon M\to\Oc\simeq T^*G$ is a trivial bundle, 
it is then straightforward to check that $\omega=\widetilde{\omega}$ 
(see the proof of Theorem~4.7 in \cite{Ba98}), and we are done. 
\end{proof}

\subsection{Induced representations of semidirect products}\label{induced-repres}
This is a classical topic for \emph{locally compact} groups (see for instance 
Ch.~5 in \cite{Ta86}). 
However in the semidirect product $M = \Fc \rtimes_\lambda G$  we are working with,  the factor $\Fc$ is generally infinite dimensional. 
Therefore in this section we shall provide a detailed construction  
of an appropriate induced representation of $M$.

In order to construct the unitary representation associated with 
the coadjoint orbit $\Oc=\Ad_M^*(M)\tilde{\delta}_{\1}$ in Proposition~\ref{orbit1} 
we need to find a real polarization of the functional $\tilde{\delta}_{\1}\in\mg^*$. 
It is not difficult to check that actually the abelian Lie algebra 
$\Fc\simeq\Fc\times\{0\}\subseteq\Fc\rtimes_{\dot{\lambda}}\gg=\mg$ 
is such a polarization, and the corresponding group is 
$\Fc\simeq\Fc\times\{\1\}\subseteq\Fc\rtimes_{\lambda}G=M$. 
Therefore the representation of the locally convex Lie group $M$ 
associated with its coadjoint orbit $\Oc$ should be the one induced 
from the  representation $\Fc\to\CC$, $\phi\mapsto\exp(\ie{\delta_{\1}}(\phi))=\ee^{\ie\phi(\1)}$. 
We now describe this induced representation in a more general setting.

Assume the setting of Proposition~\ref{top3.6} with 
$G$ an arbitrary topological group and $\Fc\hookrightarrow \Bc_{\mathbb R}(G)$ 
which is invariant under the left translation operators, 
and denote 
$M:=\Fc\rtimes_\lambda G$. 
Recall that the multiplication and the inversion in the topological group $M$ 
are defined by the equations 
$$(\phi_1,g_1)(\phi_2,g_2)=(\phi_1+\lambda_{g_1}\phi_2,g_1g_2)\text{ and }
(\phi,g)^{-1}=(-\lambda_{g^{-1}}\phi,g^{-1}) $$
respectively. 
There exist the embeddings of topological groups 
$\Fc\hookrightarrow M$, $\phi\mapsto (\phi,\1)$, 
and $G\hookrightarrow M$, $g\mapsto (0,g)$, 
and the property 
\begin{equation}\label{induction_star}
(\forall\,(\phi,g)\in M)\quad (\phi,g)=(0,g)(\lambda_{g^{-1}}\phi,\1)
\end{equation}
shows that every element in the semi-direct product $M=\Fc\rtimes_\lambda G$ 
can be uniquely written as a product of elements in the images 
of $G$ and $\Fc$ into $M$. 

Now let $u_0\colon\Fc\to{\mathbb R}$ be a linear continuous functional 
and define $\pi_0\colon\Fc\to{\TT}$, $\phi\mapsto\ee^{\ie u_0(\phi)}$, 
which is a character of the abelian topological group $(\Fc,+)$. 
We also define 
$$M\times_{\Fc}{\CC}:=(M\times\CC)/\sim$$
where $\sim$ is the equivalence relation on $M\times{\CC}$ 
defined by 
\begin{equation}\label{induction_2star}
\bigl(m(\phi,\1),z\bigr)\sim\bigl(m,\pi_0(\phi)z\bigr)
\quad
\text{ whenever }m\in M,\,\phi\in\Fc,\text{ and }z\in{\CC}. 
\end{equation}
We are going to denote by $[(m,z)]$ the equivalence class of 
any $(m,z)\in M\times{\CC}$. 
Note that there exists a natural homeomorphism 
$$M/\Fc\to G,\quad (\phi,g)\Fc\mapsto g$$
(this map is well defined because of \eqref{induction_star}) 
and a continuous surjection 
$$\Pi\colon M\times_{\Fc}{\CC}\to M/\Fc,\quad [(m,z)]\mapsto m\Fc,$$
which is actually a locally trivial bundle with the fiber~${\CC}$. 

There exists a bijective correspondence between the sections 
$\sigma\colon M/\Fc\to M\times_{\Fc}{\CC}$ 
(that is, functions satisfying $\Pi\circ\sigma=\id_{M/\Fc}$) 
and the functions $\widetilde{\sigma}\colon G\to{\CC}$. 
This correspondence is defined by 
\begin{equation}\label{induction_3star}
(\forall x\in G)\quad 
\sigma\bigl((0,x)\Fc\bigr)=\bigl[\bigl((0,x),\widetilde{\sigma}(x)\bigr)\bigr]. 
\end{equation}
Let us denote by $\Gamma_{\Borel}(M/\Fc,M\times_{\Fc}{\CC})$ 
the space of Borel measurable sections, 
so that there exists a linear isomorphism from this space onto 
the space of complex-valued, Borel measurable functions on~$G$, 
\begin{equation}\label{induction_4star}
\Gamma_{\Borel}(M/\Fc,M\times_{\Fc}{\CC})\to\Bc_{\CC}(G),
\quad 
\sigma\mapsto\widetilde{\sigma}. 
\end{equation}
The representation $\pi:=\Ind_{\Fc}^M(\pi_0)$ 
of $M$ induced by $\pi_0\colon\Fc\to{\TT}$ 
is $\pi\colon M\to\End(\Gamma_{\Borel}(M/\Fc,M\times_{\Fc}{\CC}))$ 
defined by 
$$(\pi(m)\sigma)(\mu)=m\sigma(m^{-1}\mu)\quad 
\text{ for }m\in M,\,\sigma\in\Gamma_{\Borel}(M/\Fc,M\times_{\Fc}{\CC}),
\text{ and }\mu\in M/\Fc. $$
We will denote again by 
$\pi\colon M\to\End(\Bc_{\CC}(G))$ 
the corresponding representation obtained 
by~\eqref{induction_4star}. 
To get a specific description of the latter representation~$\pi$, 
note that for every $\phi\in\Fc$ and $g,x\in G$ we have 
$$\begin{aligned}
(\pi(\phi,g)\sigma\bigl((0,x)\Fc\bigr)
&=(\phi,g)\sigma\bigl((\phi,g)^{-1}(0,x)\Fc\bigr) \\
&=(\phi,g)\sigma\bigl((-\lambda_{g^{-1}}\phi,g^{-1})(0,x)\Fc\bigr) \\
&=(\phi,g)\sigma\bigl((-\lambda_{g^{-1}}\phi,g^{-1}x)\Fc\bigr) \\
&=(\phi,g)\sigma\bigl((0,g^{-1}x)\Fc\bigr) 
 \qquad\qquad\qquad\qquad\qquad\qquad\qquad\text{(by~\eqref{induction_star})}\\
&=(\phi,g)\bigl[\bigl((0,g^{-1}x),\widetilde{\sigma}(g^{-1}x)\bigr)\bigr] 
 \qquad\qquad\qquad\qquad\qquad\;\;\text{(by~\eqref{induction_3star})}\\
&=\bigl[\bigl((\phi,g)(0,g^{-1}x),\widetilde{\sigma}(g^{-1}x)\bigr)\bigr] \\
&=\bigl[\bigl((\phi,x),\widetilde{\sigma}(g^{-1}x)\bigr)\bigr] \\
&=\bigl[\bigl((0,x)(\lambda_{x^{-1}}\phi,\1),  
    \widetilde{\sigma}(g^{-1}x)\bigr)\bigr] 
 \qquad\qquad\qquad\qquad\qquad\text{(by~\eqref{induction_star})}\\
&=\bigl[\bigl((0,x),
   \pi_0(\lambda_{x^{-1}}\phi)\widetilde{\sigma}(g^{-1}x)\bigr)\bigr]
  \qquad\qquad\qquad\qquad\qquad\text{(by~\eqref{induction_2star})}
\end{aligned}$$
whence by \eqref{induction_3star} again we get 
$$(\pi(\phi,g)\widetilde{\sigma})(x)
=\pi_0(\lambda_{x^{-1}}\phi)\widetilde{\sigma}(g^{-1}x)
\text{ for }g,x\in G,\, \phi\in\Fc,\text{ and }
\widetilde{\sigma}\in\Bc_{\CC}(G).
$$
For instance, if $u_0=\delta_{\1}\colon\Fc\to\RR$, $\phi\mapsto\phi(\1)$, 
then we get 
$$\pi:=\pi_{\1}\colon M=\Fc\rtimes_\lambda G\to\End(\Bc_{\CC}(G)), \quad 
(\pi_{\1}(\phi,g)\widetilde{\sigma})(x)
=\ee^{\ie\phi(x)}\widetilde{\sigma}(g^{-1}x).
$$
If we define $U\colon\Bc_{\CC}(G)\to\Bc_{\CC}(G)$, 
$(U\widetilde{\sigma})(x)=\widetilde{\sigma}(x^{-1})$, then we get 
the equivalent representation $U\pi_{\1}(\cdot)U^{-1}$ 
with the specific expression 
$$(U\pi_{\1}(\phi,g)U^{-1}\widetilde{\sigma})(x)
=(\pi_{\1}(\phi,g)U^{-1}\widetilde{\sigma})(x^{-1})
=\ee^{\ie\phi(x^{-1})}(U^{-1}\widetilde{\sigma})(g^{-1}x^{-1})
=\ee^{\ie\phi(x^{-1})}\widetilde{\sigma}(xg)
$$
for $g,x\in G$, $\phi\in\Fc$, and 
$\widetilde{\sigma}\in\Bc_{\CC}(G)$.

\section{Magnetic  preduals of the coadjoint orbit $\Oc$} 

\subsection{Auxiliary properties of nilpotent Lie algebras}

\begin{definition}\label{nilp1}
\normalfont
Let $\gg$ be a nilpotent finite-dimensional real Lie algebra 
of dimension~$\ge1$ and define $\gg_0:=\gg$ and 
$$(\forall k\ge1)\quad 
\gg_k=\spa\{[X_k,\dots,[X_1,X_0]\dots]\mid X_0,X_1,\dots,X_k\in\gg\}.
$$
Then $\gg_0\supseteq\gg_1\supseteq\gg_2\supseteq\cdots$ 
and, since $\gg$ is a nilpotent Lie algebra, 
there exists $n\ge0$ with $\gg_n\ne\{0\}=\gg_{n+1}$. 
The number $n\ge0$ is called the \emph{nilpotency index} of $\gg$. 

Note that $[\gg,\gg_n]=\gg_{n+1}=\{0\}$, 
hence $\gg_n$ is contained in the center of~$\gg$. 
In particular, $\gg_n$ is an ideal of $\gg_n$ and then there exists 
a natural Lie bracket on $\gg/\gg_n$ 
which makes the quotient map $q\colon\gg\to\gg/\gg_n$ 
into a homomorphism of Lie algebras. 
It is also easily seen that $\gg/\gg_n$ 
is a nilpotent Lie algebra whose nilpotency index is $n-1$, 
provided that $n\ge1$. 
\qed
\end{definition}

\begin{proposition}\label{nilp2}
If $\gg$ is a nilpotent finite-dimensional real Lie algebra, 
then for every $V\in\gg$ the mapping 
$$\Psi_{\gg,V}\colon\gg\to\gg,\quad Y\mapsto\int\limits_0^1Y\ast(sV)\,\de s$$
is a polynomial diffeomorphism whose inverse is also polynomial 
and which preserves the Lebesgue measure. 
\end{proposition}

\begin{proof}
Recall that the multiplication $\ast$ defined by 
the Baker-Campbell-Hausdorff (BCH) formula is a polynomial mapping 
in the case of the nilpotent Lie algebras, and therefore the mapping in 
the statement is polynomial. 
To prove the other properties 
we shall proceed by induction on the nilpotency index 
of the Lie algebra under consideration. 

If the nilpotency index of $\gg$ is $0$, then this algebra is abelian, 
so the BCH multiplication $\ast$ reduces to the vector sum. 
Then for every $V\in\gg$ we have 
$$(\forall Y\in\gg)\quad 
\Psi_{\gg,V}(Y)=\int\limits_0^1Y+sV\,\de s=Y+\frac{1}{2}V,$$
which clearly has the properties we wish for. 

Now let $n\ge1$ and assume that the assertion holds for the Lie algebras 
of nilpotency index~$<n$. 
Let $\gg$ be a nilpotent Lie algebra with $\gg_n\ne\{0\}=\gg_{n+1}$ 
(see the notation in Definition~\ref{nilp1}) and take $V\in\gg$ arbitrary. 
To show that the mapping $\Psi_{\gg,V}\colon\gg\to\gg$ is injective, 
let $Y_1,Y_2\in\gg$ such that $\Psi_{\gg,V}(Y_1)=\Psi_{\gg,V}(Y_2)$. 
If we transform both sides of the latter equation by 
the Lie algebra homomorphism $q\colon\gg\to\gg/\gg_n$ 
which preserves the BCH multiplication, 
then we get 
$\int\limits_0^1q(Y_1)\ast(sq(V))\,\de s=\int\limits_0^1q(Y_2)\ast(sq(V))\,\de s$. 
Since the mapping $\Psi_{\gg/\gg_n,q(V)}\colon\gg/\gg_n\to\gg/\gg_n$ 
is injective by the induction hypothesis, it follows that 
$q(Y_1)=q(Y_2)$, that is, $Y_0:=Y_1-Y_2\in\Ker q=\gg_n$. 
Then 
$$\begin{aligned}
\Psi_{\gg,V}(Y_1)
&=\Psi_{\gg,V}(Y_2+Y_0)
=\int\limits_0^1(Y_0+Y_2)\ast(sV)\,\de s 
=\int\limits_0^1Y_0+(Y_2\ast(sV))\,\de s \\
&=Y_0+\int\limits_0^1Y_2\ast(sV)\,\de s 
=Y_0+\Psi_{\gg,V}(Y_2),
\end{aligned}$$
so the assumption $\Psi_{\gg,V}(Y_1)=\Psi_{\gg,V}(Y_2)$ 
implies $Y_0=0$, whence $Y_1=Y_2$. 
We note that the above equalities follow by using 
the definition of the BCH multiplication $\ast$ 
along with the fact that $Y_0\in\gg_n$, hence $[Y_0,\gg]=\{0\}$. 

It remains to check that the mapping $\Psi_{\gg,V}\colon\gg\to\gg$ 
is surjective and its inverse is polynomial. 
For that purpose let $\iota\colon\gg/\gg_n\to\gg$ 
be any linear mapping satisfying $q\circ\iota=\id_{\gg/\gg_n}$. 
(So $\iota$ can be any linear isomorphism of $\gg/\gg_n$ 
onto a linear complement of $\gg_n$ in~$\gg$.) 
Denote 
\begin{equation}\label{nilp2_star}
(\forall Z\in\gg)\quad 
\Delta(Z):=Z-\int\limits_0^1\iota(\Phi(q(Z)))\ast(sV)\,\de s,
\end{equation}
where $\Phi:=(\Psi_{\gg/\gg_n,q(V)})^{-1}\colon\gg/\gg_n\to\gg/\gg_n$ 
is a polynomial map which exists because of the induction hypothesis. 
Note that for every $Z\in\gg$ we have 
$$\begin{aligned}
q(\Delta(Z))
&=q(Z)-q\Bigl(\int\limits_0^1\iota(\Phi(q(Z)))\ast(sV)\,\de s\Bigr) 
=q(Z)-\int\limits_0^1q(\iota(\Phi(q(Z))))\ast q(sV)\,\de s \\
&=q(Z)-\int\limits_0^1\Phi(q(Z))\ast(sq(V))\,\de s 
=0,
\end{aligned}$$
where we used the equality $q\circ\iota=\id_{\gg/\gg_n}$ 
and again the fact that $q\colon\gg\to\gg/\gg_n$ 
is a Lie algebra homomorphism hence preserves the BCH multiplications. 
Since $\Ker q=\gg_n$ and $[\gg,\gg_n]=\gg_{n+1}=\{0\}$, 
we get 
$$(\forall Z\in\gg)\quad [\Delta(Z),\gg]=\{0\}.$$
We can use this property to see that 
(as in the above proof of the fact that $\Psi_{\gg,V}$ is injective) 
we have for every $Z\in\gg$, 
$$\begin{aligned}
Z
&=\Delta(Z)+\int\limits_0^1\iota(\Phi(q(Z)))\ast(sV)\,\de s 
=\int\limits_0^1\Delta(Z)+(\iota\circ\Phi\circ q)(Z)\ast(sV)\,\de s \\
&=\int\limits_0^1(\Delta(Z)+(\iota\circ\Phi\circ q)(Z))\ast(sV)\,\de s 
=\Psi_{\gg,V}(\Delta(Z)+(\iota\circ\Phi\circ q)(Z)).
\end{aligned}$$
This shows that the mapping $\Psi_{\gg,V}\colon\gg\to\gg$ 
is indeed surjective and 
\begin{equation}\label{nilp2_2star}
(\forall Z\in\gg)\quad 
(\Psi_{\gg,V})^{-1}(Z)=\Delta(Z)+(\iota\circ\Phi\circ q)(Z). 
\end{equation}
To conclude the proof, just recall that 
$\Phi=(\Psi_{\gg/\gg_n,q(V)})^{-1}\colon\gg/\gg_n\to\gg/\gg_n$ 
is a polynomial map by the induction hypothesis, 
while the BCH multiplication  
is a polynomial mapping on every nilpotent Lie algebra. 
Since both $\iota$ and $q$ are linear, 
it follows by \eqref{nilp2_star} that $\Delta\colon\gg\to\gg$ 
is a polynomial mapping, and then \eqref{nilp2_2star} 
shows that so is $(\Psi_{\gg,V})^{-1}\colon\gg\to\gg$. 

As regards the measure-preserving property, it will be enough to show that 
for an arbitrary $Y_0\in\gg$ the differential $(\Psi_{\gg,V})'_{Y_0}\colon\gg\to\gg$ 
is a linear map whose determinant is equal to~$1$. 
To this end note that for every $Y\in\gg_n$ we have $[\gg,Y]=\{0\}$ 
hence $\Psi_{\gg,V}(Y)=\int\limits_0^1 Y+sV\,\de s=Y+\frac{1}{2}V$, 
which implies that  $\gg_n$ is invariant under the differential $(\Psi_{\gg,V})'_{Y_0}$. 
Actually, the latter map restricted to $\gg_n$ is equal to the identity map on $\gg_n$, 
and in particular the determinant of that restriction is equal to~$1$. 
On the other hand, as above in the proof of injectivity of $\Psi_{\gg,V}$,  
we get $q\circ\Psi_{\gg,V}=\Psi_{\gg/\gg_n,q(V)}\circ q$. 
By differentiating this equality at $Y_0\in\gg$ and taking into account 
that $q\colon\gg\to\gg/\gg_n$ is a linear map, we get 
$q\circ(\Psi_{\gg,V})'_{Y_0}=(\Psi_{\gg/\gg_n,q(V)})'_{q(Y_0)}\circ q$, and 
then we get the following commutative diagram 
$$\begin{CD}
 \gg_n @>>> \gg @>{q}>> \gg/\gg_n  \\
 @VV{\id_{\gg_n}}V @VV{(\Psi_{\gg,V})'_{Y_0}}V @VV{(\Psi_{\gg/\gg_n,q(V)})'_{q(Y_0)}}V \\
 \gg_n @>>> \gg @>{q}>> \gg/\gg_n  \\
\end{CD}$$
whose rows are short exact sequences. 
Since the determinant of $(\Psi_{\gg/\gg_n,q(V)})'_{q(Y_0)}$ is equal to~$1$ 
by the induction hypothesis, it follows that 
the determinant of the middle vertical arrow is also equal to~1. 
This completes the induction step and the proof. 
\end{proof}

\subsection{Magnetic preduals and global coordinates for $\Oc$}\label{magn_quant}

We shall work in the following setting:
\begin{enumerate}
\item The simply connected nilpotent Lie group $G$ is identified with its Lie algebra $\gg$ 
by means of the exponential map and $\ast$ denotes the Baker-Campbell-Hausdorff multiplication on~$\gg$. 
\item  We denote by 
$\langle\cdot,\cdot\rangle\colon\gg^*\times\gg\to\RR$ the canonical duality pairing. 
\item 
We also denote by $\Fc$ an admissible space of functions on $\gg$ 
which contains both $\gg^*$ and the constant functions. 
As usual, we denote by  $M=\Fc\rtimes_\lambda\gg$ the corresponding semidirect product of groups, 
which is a locally convex Lie group with the Lie algebra $T_{(0,0)}M=\mg=\Fc\rtimes_{\dot{\lambda}}\gg$. 
Here we shall distinguish $\mg$ from the set $\Lie(M)$ of one-parameter subgroups in $M$. 
\item The \emph{magnetic potential} $A\in\Omega^1(\gg)$ is 
a smooth differential 1-form whose coefficients belong to~$\Fc$. 
That is, $A\colon\gg\to\gg^*$, $X\mapsto A_X:=A(X)$, is a smooth mapping such that 
for every $X\in\gg$ the function $Y\mapsto\langle A_Y,(R_Y)'_0X\rangle$ 
belongs to~$\Fc$, where $R_Y\colon\gg\to\gg$, $R_Y(W)=W\ast Y$. 
\item The \emph{magnetic field} is the 2-form $B=dA\in\Omega^2(\gg)$. 
Hence $B$ is a smooth mapping $X\mapsto B_X$ from~$\gg$ into the space of 
all skew-symetric bilinear functionals on $\gg$ 
such that 
$$(\forall X,X_1,X_2\in\gg)\quad 
B_X(X_1,X_2)=\langle A'_X(X_1),X_2\rangle-\langle A'_X(X_2),X_1\rangle.$$
\end{enumerate}

\begin{proposition}\label{pullback}
For every $\phi\in\Fc$ and $X\in\gg$ define $\bar\theta_0^A(\phi,X)\in\Fc$ by 
$$(\forall Y\in\gg)\quad (\bar\theta_0^A(\phi,X)) (Y)=\phi(Y)+\langle A_Y,(R_Y)'_0X\rangle,$$
and then consider the continuous linear mapping 
$$\bar\theta^A\colon\Fc\rtimes_{\dot\lambda}\gg=\mg\to\mg,\quad 
\bar\theta(\phi,X)=(\bar\theta_0^A(\phi,X),X)$$
and the differential 2-forms 
$$\bar\omega\in\Omega^2(\Fc\times\gg),\quad 
\bar\omega_{(\phi_0,X_0)}((\phi_1,X_1),(\phi_2,X_2))=
(\dot{\lambda}(X_1)\phi_2-\dot{\lambda}(X_2)\phi_1)(X_0)
$$
and 
$$\bar B\in\Omega^2(\Fc\times\gg),\quad 
\bar B_{(\phi_0,X_0)}((\phi_1,X_1),(\phi_2,X_2))=B_{X_0}(X_1,X_2).$$
Then the following assertions hold: 
\begin{enumerate}
\item The operator $\bar\theta^A\colon\mg\to\mg$ is invertible and 
$(\bar\theta^A)^{-1}=\bar\theta^{-A}$. 
\item If $\gg$ is a two-step nilpotent Lie algebra, 
then 
$$\bar\omega\in\Omega^2(\Fc\times\gg),\quad 
\bar\omega_{(\phi_0,X_0)}((\phi_1,X_1),(\phi_2,X_2))=(\phi_2)'_{X_0}(X_1)-(\phi_1)'_{X_0}(X_2).$$
Moreover, $d\bar\omega=0$ and $(\bar\theta^A)^*(\bar\omega)=\bar\omega+\bar B$. 
\end{enumerate}
\end{proposition}

\begin{proof} 
The first assertion is easily seen. 
For the second assertion, note that if $\gg$ is two-step nilpotent, 
then $(R_Y)'_0=\id_{\gg}$ for every $Y\in\gg$, 
hence the specific expression of $\bar\omega$ follows by~\eqref{lambda_diff}. 
If we regard $\bar\omega$ as a mapping from $\Fc\times\gg$ into the skew-symmmetric bilinear functionals 
on $\Fc\times\gg$, then we may dfferentiate it as such and we get 
$$\begin{aligned}
d\bar\omega_{(\phi_0,X_0)}((\phi_1,X_1),(\phi_2,X_2),(\phi_3,X_3))
=&\bar\omega'_{(\phi_0,X_0)}(\phi_1,X_1)((\phi_2,X_2),(\phi_3,X_3)) \\
&-\bar\omega'_{(\phi_0,X_0)}(\phi_2,X_2)((\phi_1,X_1),(\phi_3,X_3)) \\
&+\bar\omega'_{(\phi_0,X_0)}(\phi_3,X_3)((\phi_1,X_1),(\phi_2,X_2)) \\
=&(\phi_3)''_{X_0}(X_1,X_2)-(\phi_2)''_{X_0}(X_1,X_3) \\
&-(\phi_3)''_{X_0}(X_2,X_1)+(\phi_1)''_{X_0}(X_2,X_3) \\
&+(\phi_2)''_{X_0}(X_3,X_1)-(\phi_1)''_{X_0}(X_3,X_2) \\
=&0
\end{aligned}
$$
since the second differentials of the smooth functions $\phi_1,\phi_2,\phi_3\in\Fc$ are symmetric. 
Further, since $\bar\theta$ is a linear map we get 
$$\begin{aligned}
\bar\theta^*(\bar\omega)_{(\phi_0,X_0)}((\phi_1,X_1),(\phi_2,X_2))
=&\bar\omega_{\bar\theta(\phi_0,X_0)}(\bar\theta(\phi_1,X_1),\bar\theta(\phi_2,X_2)) \\
=&\bar\omega_{(\phi_0+\langle A(\cdot),X_0\rangle,X_0)}
 ((\phi_1+\langle A(\cdot),X_1\rangle,X_1),(\phi_2+\langle A(\cdot),X_2\rangle,X_2)) \\
=&(\phi_2)'_{X_0}(X_1)+\langle A'_{X_0}(X_1),X_2\rangle 
-(\phi_1)'_{X_0}(X_2)-\langle A'_{X_0}(X_2),X_1\rangle  \\
=&\bar\omega_{(\phi_0,X_0)}((\phi_1,X_1),(\phi_2,X_2))+\bar B_{(\phi_0,X_0)}((\phi_1,X_1),(\phi_2,X_2)), 
\end{aligned}$$
and this completes the proof. 
\end{proof}

\begin{definition}\label{predual}
\normalfont
Assume the notation introduced in Proposition~\ref{pullback}. 
The set 
$$\Oc_*=\Oc_*^A:=\{(\bar\theta_0^A(\xi,X),X)\mid X\in\gg,\xi\in\gg^*\}
\subseteq\Fc\rtimes_{\dot\lambda}\gg=\mg
$$
will be called the \emph{magnetic predual of the coadjoint orbit~$\Oc$} 
(associated with the magnetic potential~$A$). 
Let 
$I\colon \gg\times\gg^*\hookrightarrow\Fc\times\gg$ 
be the natural embedding $I(X,\xi)=(\xi,X)$. 
Then the mapping 
\begin{equation}\label{magn_isom}
\bar\theta^A\circ I\colon\gg\times\gg^*\to\Oc_*
\end{equation}
is a linear isomorphism which (by Proposition~\ref{pullback}) 
takes the canonical symplectic structure of $\gg\times\gg^*$ 
to a certain symplectic structure on $\Oc_*$, 
which will be called the \emph{natural symplectic structure of the magnetic predual~$\Oc_*$}. 
Thus~\eqref{magn_isom} is an isomorphism of symplectic vector spaces.
\qed
\end{definition}

\begin{remark}\label{predual_conj}
\normalfont
The magnetic predual $\Oc_*^A$ essentially depends only on 
the magnetic field $B=dA$. 
Specifically, if $A_1,A_2\in\Omega^1(\gg)$ are magnetic potentials 
then there exists $m_0=(\phi_0,X_0)\in M$ such that 
$\bar\theta^{A_1}=\Ad_M(m)\circ\bar\theta^{A_2}$ 
if and only if $dA_1=dA_2$. 
This follows by using Remark~\ref{semidirect_rem}\eqref{semidirect_rem_item2}. 
\qed
\end{remark}

In the following statement we need some notation from Propositions \ref{top3.6}~and~\ref{orbit1}. 
Thus, $\delta_0\colon\Fc\to\RR$ is the functional $\phi\mapsto\phi(0)$. 

\begin{proposition}\label{para}
Let us define 
$$\theta_0\colon\gg\times\gg^*\to\Fc,\quad (\theta_0(X,\xi))(y)=\langle\xi,Y\rangle+\langle A_Y,(R_Y)'_0X\rangle$$
and 
$$\theta\colon\gg\times\gg^*\to\Lie(M),\quad \theta(X,\xi)=Z_{\theta_0(X,\xi),X}.$$ 
Then the mapping 
$$\Pi_{\tilde{\delta}_0}\circ\exp_M\circ\theta\colon\gg\times\gg^*\to\mg^*,\quad  (X,\xi)\mapsto\Ad_M^*(\exp_M(\theta(X,\xi)))\tilde{\delta}_0$$
is a diffeomorphism of $\gg\times\gg^*$ onto the coadjoint orbit $\Oc$ of 
$\tilde{\delta}_0=(\delta_0,0)\in\Fc^*\times\gg^*=\mg^*$. 
\end{proposition}

\begin{proof}
Let $M_{\tilde{\delta}_0}$ be the coadjoint isotropy group at $\tilde{\delta}_0\in\mg^*$. 
To prove that $\Phi\colon\gg\times\gg^*\to\mg^*$ is a bijection onto 
$\Ad_M^*(M)\tilde{\delta}_0\simeq M/M_{\tilde{\delta}_0}$ it is necessary and sufficient 
to see that the following assertions hold: 
\begin{enumerate}
\item\label{exp_theta_inj} 
The mapping $\gg\times\gg^*\to M$, $(X,\xi)\mapsto\exp_M(\theta(X,\xi))$ is injective. 
\item\label{exp_theta_main}
The multiplication mapping 
\begin{equation}\label{exp_theta_mult}
\exp_M(\theta(\gg\times\gg^*))\times M_{\tilde{\delta}_0}\to M
\end{equation}
is bijective and additionally, if $m_1,m_2\in\exp_M(\theta(\gg\times\gg^*))$ satisfy $m_1\in m_2M_{\tilde{\delta}_0}$, 
then necessarily $m_1=m_2$. 
\end{enumerate}
In order to prove these assertions we shall use fact that by Proposition~\ref{top3.6} we have 
\begin{equation}\label{exp_theta}
(\forall(X,\xi)\in\gg\times\gg^*)\quad \exp_M(\theta(X,\xi))=(\alpha(X,\xi),X)\in\Fc\rtimes_\lambda\gg=M.  
\end{equation}
Here the function $\alpha(X,\xi)\in\Fc$ at an arbitrary point $Y\in\gg$ 
can be computed in the following way: 
$$\begin{aligned}
(\alpha(X,\xi)(Y)
&=\int\limits_0^1(\lambda_{sX}(\theta_0(X,\xi))(Y)\de s 
=\int\limits_0^1(\theta_0(X,\xi))((-sX)\ast Y)\de s \\
&=\int\limits_0^1\xi((-sX)\ast Y)\de s+\int\limits_0^1\langle A((-sX)\ast Y),(R_{(-sX)\ast Y})'_0 X\rangle\de s 
\end{aligned}$$
By using the notation introduced in Proposition~\ref{nilp2} we get 
\begin{equation}\label{para_eq6}
\begin{aligned}
(\alpha(X,\xi)(Y)
&=\langle\xi,\int\limits_0^1(-sX)\ast Y\de s\rangle
  +\langle\int\limits_0^1 A((-sX)\ast Y)\de s,(R_{(-sX)\ast Y})'_0 X\rangle    \\
&=-\langle\xi,\Psi_{\gg,X}(-Y)\rangle
  +\langle\int\limits_0^1 A((-sX)\ast Y)\de s,(R_{(-sX)\ast Y})'_0 X\rangle.
\end{aligned}
\end{equation}
Now, to prove assertion~\eqref{exp_theta_inj}, just note that 
if $\exp_M(\theta(X_1,\xi_1))=\exp_M(\theta(X_2,\xi_2))$, 
then by~\eqref{exp_theta} we get $X_1=X_2=:X$ and $\alpha(X,\xi_1)=\alpha(X,\xi_2)$. 
Then by~\eqref{para_eq6} we get $\xi_1\circ\Psi_{\gg,X}=\xi_2\circ\Psi_{\gg,X}$. 
Since $\Psi_{\gg,X}\colon\gg\to\gg$ is a diffeomorphism by Proposition~\ref{nilp2}, 
it follows that $\xi_1=\xi_2$. 

We now proceed to proving the above assertion~\eqref{exp_theta_main}. 
To prove the second part of that assertion, let us assume that 
$\exp_M(\theta(X_1,\xi_1))\in\exp_M(\theta(X_2,\xi_2))M_{\tilde{\delta}_0}$. 
It then follows by \eqref{isotropy} and \eqref{exp_theta} that 
there exists $\phi\in\Fc$ such that $\phi'_{0}=0$ and 
$(\alpha(X_1,\xi_1),X_1)=(\alpha(X_2,\xi_2),X_2)(\phi,0)$. 
Thence $X_1=X_2=:X$ and $\alpha(X,\xi_1)=\alpha(X,\xi_2)+\lambda_X\phi$, 
so by~\eqref{para_eq6} we get 
$\langle\xi_2-\xi_1,\Psi_{\gg,X}(-Y)\rangle=\phi((-X)\ast Y)$ for every $Y\in\gg$. 
By means of the change of variable $(-X)\ast Y=-W$ we have $W\ast(-X)=-Y$, 
and then $\langle\xi_2-\xi_1,\Psi_{\gg,X}(W\ast(-X))\rangle=\phi(-W)$ for every $W\in\gg$. 
Now note that 
$$\Psi_{\gg,X}(W\ast(-X))=\int\limits_0^1W\ast(-X)\ast(sX)\de s
=\int\limits_0^1W\ast((1-s)X)\de s=\int\limits_0^1W\ast(sX)\de s=\Psi_{\gg,X}(W)$$
hence $\langle\xi_2-\xi_1,\Psi_{\gg,X}(W)\rangle=\phi(-W)$ for every $W\in\gg$. 
By differentiating the latter equation at $W=0$ we get 
$(\xi_2-\xi_1)\circ(\Psi_{\gg,X})'_0=\phi'_0=0$. 
Now recall that $\Psi_{\gg,X}\colon\gg\to\gg$ is a diffeomorphism by Proposition~\ref{nilp2}, 
hence $(\Psi_{\gg,X})'_0\colon\gg\to\gg$ is a linear isomorphism, and then $\xi_2-\xi_1=0$. 

This proves the second part of assertion~\eqref{exp_theta_main} 
which in particular shows that the multiplication mapping~\eqref{exp_theta_mult} is injective. 
To prove that that mapping is surjective as well, 
let $(\phi,X)\in M$ arbitrary. 
It follows by \eqref{isotropy} and \eqref{exp_theta} again that 
it will be enough to find $\xi\in\gg^*$ and $\psi\in\Fc$ such that $\psi'_0=0$ and 
$(\alpha(X,\xi),X)(\psi,0)=(\phi,X)$. 
The latter equation is equivalent to $\alpha(X,\xi)+\lambda_X\psi=\phi$, 
that is, $\lambda_{-X}(\alpha(X,\xi))+\psi=\lambda_{-X}\psi$, 
whence by \eqref{para_eq6} we get
$$(\forall Y\in\gg)\quad 
\langle\xi,\int\limits_0^1(-sX)\ast X\ast Y\de s\rangle
+\langle\int\limits_0^1 A((-sX)\ast X\ast Y)\de s,(R_{(-sX)\ast X\ast Y})'_0 X\rangle 
+\psi(Y)=\phi(X\ast Y).$$
Since $(-sX)\ast X=(1-s)X$, the above equation is further equivalent to 
\begin{equation}\label{exp_theta_surj}
(\forall Y\in\gg)\quad 
\langle\xi,\int\limits_0^1(sX)\ast Y\de s\rangle+\langle\int\limits_0^1 A((sX)\ast Y)\de s,(R_{(sX)\ast Y})'_0 X\rangle 
+\psi(Y)=\phi(X\ast Y).
\end{equation}
Since the mapping $Y\mapsto\int\limits_0^1(sX)\ast Y\de s=-\Psi_{\gg,-X}(-Y)$
is a diffeomorphism by Proposition~\ref{nilp2}, it folows that its differential at $Y=0$ 
is a linear isomorphism on $\gg$. 
Now by differentiating~\eqref{exp_theta_surj} at $Y=0$ and using the condition $\psi'_0=0$, 
we see that $\xi\in\gg^*$ can be uniquely determined in terms of the given function $\varphi\in\Fc$. 
Then we just have to solve equation~\eqref{exp_theta_surj} for $\psi$. 
This completes the proof of the fact that the multiplication mapping~\eqref{exp_theta_mult} 
is surjective. 

We now know that the mapping 
$\Pi_{\tilde{\delta}_0}\circ\exp_M\circ\theta\colon\gg\times\gg^*\to\Oc$ 
in the statement is a bijection. 
To see that it is actually a diffeomorphism, firstly note that 
$\exp_M\circ\theta\colon \gg\times\gg^*\to M$ is smooth as an easy consequence 
of \eqref{exp_theta}~and~\eqref{para_eq6}, 
and then $\Pi_{\tilde{\delta}_0}\circ\exp_M\circ\theta$ is smooth. 
To prove that its inverse is also smooth, we just have to use the fact 
that the solution $\xi$ of \eqref{exp_theta_surj} depends smoothly on the data $\phi\in\Fc$ 
(as a direct consequence of our way to solve equation~\eqref{exp_theta_surj}). 
\end{proof}

\begin{corollary}\label{predual_diffeo}
Let $\tilde{\delta}_0=(\delta_0,0)\in\mg^*$. 
The mapping 
$$\Ad_M^*(\exp_M(\cdot))\tilde{\delta}_0\colon\Oc_*\to\Oc$$
is a diffeomorphism.
\end{corollary}

\begin{proof}
Use Propositions \ref{pullback}~and~\ref{para}.
\end{proof}

\section{Magnetic Weyl calculus on Lie groups}

\subsection{Localized Weyl calculus}\label{localized}
In this subsection we sketch a general setting, 
inspired by \cite{An69} and \cite{An72}, for the Weyl calculus 
associated with continuous representations of any topological groups, 
which may be infinite-dimensional Lie groups. 
We shall apply this construction in the next subsection 
in the case of a semidirect product $M=\Fc\rtimes_{\lambda}G$, 
where $\Fc$ is a certain function space on the nilpotent Lie group~$G$.

\begin{definition}\label{top4}
\normalfont
Let $M$ be a topological group and $\pi\colon M\to\Bc(\Yc)$ 
a \so-continuous, uniformly bounded representation 
on the complex separable Banach space $\Yc$. 
Assume the setting defined by the following data: 
\begin{enumerate} 
\item 
a duality pairing $\langle\cdot,\cdot\rangle\colon\Xi^*\times\Xi\to{\mathbb R}$ 
between two real finite-dimensional vector spaces $\Xi$ and $\Xi^*$; 
\item 
a map $\theta\colon \Xi\to\Lie(M)$ which is measurable 
with respect to the natural Borel structures of $\Xi$ and~$\Lie(M)$. 
\end{enumerate}
Denote by 
$$\widehat{\cdot}
\colon L^1(\Xi)\to L^\infty(\Xi^*), \quad 
b(\cdot)\mapsto\widehat{b}(\cdot)
=\int\limits_{\Xi}\ee^{-\ie\langle\cdot,x\rangle}b(x)\,\de x $$
the \emph{Fourier transform with respect to the duality} 
$\langle\cdot,\cdot\rangle$, 
and the inverse Fourier transform 
$$\check{\cdot}
\colon L^1(\Xi^*)\to L^\infty(\Xi), \quad 
a(\cdot)\mapsto\check{a}(\cdot)
=\int\limits_{\Xi^*}\ee^{\ie\langle\xi,\cdot\rangle}a(\xi)\,\de\xi  $$
where the Lebesgue measures on $\Xi$ and $\Xi^*$ are 
suitably normalized. 

Then the corresponding \emph{localized Weyl calculus for $\pi$ along~$\theta$} 
is defined by 
\begin{equation}\label{weyl_loc}
\Op^\theta\colon\widehat{L^1(\Xi)}\to\Bc(\Yc),\quad 
\Op^\theta(a)y
=\int\limits_{\Xi} \check{a}(\xi)\pi(\exp_S(\theta(\xi)))y\,\de\xi
\text{ for }y\in\Yc\text{ and }a\in\widehat{L^1(\Xi)},
\end{equation}
where we use Bochner integrals of $\Yc$-valued functions.
\qed
\end{definition}

\begin{remark}\label{top5}
\normalfont
In the setting of Definition~\ref{top4} we note the following: 
\begin{enumerate}
\item\label{top5_item1}
We need the Banach space $\Yc$ to be separable in order to define the Bochner integral. 
Instead, we could have assumed $\Yc$ a reflexive Banach space (for instance a Hilbert space) 
and defined $W^\theta(f)\in\Bc(\Yc)$ as a weakly convergent integral.
\item\label{top5_item2}
It follows by \eqref{nat_var} that 
\begin{equation}\label{gener}
\Op^\theta(a)y
=\int\limits_{\Xi^*}\check{a}(\xi)\exp\bigl(\Lie(\pi)(\theta(\xi))\bigr)y\,  \de\xi
\end{equation}
for $y\in\Yc$ and $a\in\widehat{L^1(\Xi)}$, 
hence the localized Weyl functional calculus 
actually depends on the mapping $\Lie(\pi)\colon\Lie(M)\to\Cc(\Yc)$, 
rather than on the representation $\pi\colon M\to\Bc(\Yc)$ itself. 
If $\Yc$ is a Hilbert space, $\pi$ is a unitary representation, 
and $\theta\colon\Xi\to\Lie(M)$ is continuous, it easily follows that~\eqref{gener} makes sense for 
every bounded continuous function $a\colon\Xi^*\to\CC$ 
whose inverse Fourier transform $\check a$ is a finite Radon measure on~$\Xi$. 
It thus follows that 
for every $\xi_0\in\Xi^*$ we get the usual functional calculus 
of the self-adjoint operator $\Lie(\pi)\theta(\xi_0)$ 
by suitably extending $\Op^\theta$ to functions of the form $\xi\mapsto b(\scalar{\xi_0}{\xi})$ 
with $b\colon\RR\to\CC$. 
\item\label{top5_item3} 
The localized Weyl functional calculus for $\pi$ along~$\theta$ 
has the following covariance property: 
If $\theta'\colon\Xi^*\to\Lie(M)$ is another measurable map such that 
there exists $m\in M$ satisfying $\Ad_M(m)\circ\theta'=\theta$, then 
\begin{equation}\label{cov}
(\forall a\in\widehat{L^1(\Xi)})\quad 
\Op^{\theta}(a)=\pi(m) \Op^{\theta'}(a) \pi(m)^{-1}. 
\end{equation}
In fact, for every $\xi\in\Xi^*$ we have 
$$\exp_M(\theta(\xi))=(\theta(\xi))(1)=(\bigl(\Ad_M(m)\bigr)\theta'(\xi))(1)
=m(\theta'(\xi)(1))m^{-1}=m\exp_M(\theta'(\xi))m^{-1},$$ 
hence $\pi(\exp_M(\theta(\xi)))=\pi(m)\pi(\exp_M(\theta'(\xi)))\pi(m)^{-1}$, 
and now \eqref{cov} follows by~\eqref{weyl_loc}.  
\end{enumerate}
\qed
\end{remark}

\subsection{Magnetic pseudo-differential calculus on nilpotent Lie groups}

We are going to specialize here the ideas of subsection~\ref{localized} 
in order to construct a magnetic Weyl calculus in the setting of subsection~\ref{magn_quant}. 
Thus $G$ is a (connected and) simply connected nilpotent Lie group with $\Lie(G)=\gg$. 
Then the exponential map $\exp_G\colon\gg \mapsto G$ is a diffeomorphism, and 
we use the notation $ \log_G = \exp_G^{-1}$.
We recall that the Haar measure on the group $G$ is taken by $\log_G$ into the Lebesque measure on 
$\gg$, consequently 
the Lebesque measure on $\gg$ is invariant under the transformations 
$Y\mapsto Y\ast X$ and $ Y \mapsto (-Y)$.

Assume  $\Fc$ an admissible space of real continuous functions on $G$, which is invariant
under the left regular action, hence the mapping 
$$ \lambda\colon G \times \Fc \to  \Fc, \quad (\lambda_g\varphi)(x) = \varphi(g^{-1}x)$$
is well defined.  
Since $\Fc$ is endowed with a topology such that $\lambda$ is continuous, 
we may consider the semidirect product $M= \Fc  \rtimes_\lambda \,G$. 
Proposition~\ref{top3.6} shows that the Lie algebra of $M$ is the semidirect product 
$\Fc \rtimes_{\dot{\lambda}} \, \gg $ and the exponential map $\exp_M$ is given by 
$$ \exp_M (\varphi, X) = \big( 
\int   \limits_0^1 \lambda_{\exp_G(sX)} \varphi\, \de s, \,  \exp_G(X)\big).
$$
We denote  the duality between $\gg$ and $\gg^*$ also by 
$$\gg^* \times \gg \ni (\xi,X)\mapsto \scalar{\xi}{X}\in \RR.$$
Then we assume that the functions  $\xi\circ\log_G$ belong to $\Fc$ for every $\xi \in \gg^*$.

We set $\Xi= \gg\times \gg^*$. 
The mapping
$$ \scalar{\cdot}{\cdot}\colon \Xi\times \Xi\to \RR, \quad 
\scalar{(X_1, \xi_1)}{ (X_2, \xi_2)}= \scalar{\xi_1}{X_2}- \scalar{\xi_2}{X_1} $$ 
defines a symplectic structure on $\Xi$. 
This is, in particular, a duality pairing, $\Xi$ being self-dual with respect 
to this pairing.  
The Fourier transform associated to $ \scalar{\cdot}{\cdot}$ is given by 
$$ (F_\Xi a)(X, \xi) =\hat a(X, \xi) = \int\limits_{\Xi} \ee^{-\ie \scalar{(X, \xi)}{ (Y, \eta)}} 
a(Y, \eta) \, \de(Y, \eta), \quad a \in L^1(\Xi).
$$
It extends to an invertible operator $\Sc'(\Xi) \to \Sc' (\Xi)$, 
$F_\Xi^{-1}= F_\Xi$ 
 and we denote 
$\check a= F_{\Xi}^{-1}a$. 
Note that if $F_\gg\colon \Sc'(\gg) \to \Sc' (\gg)$ is the Fourier transform 
associated to the duality between $\gg$ and $\gg^*$ 
(normalized such that it is unitary $L^2(\gg) \to L^2(\gg^*)$)
then  
\begin{equation}\label{fourier:1}
 F_\Xi = \iota ^* (F_\gg \otimes F_\gg^{-1}) = (F_\gg^{-1} \otimes F_\gg) (\iota^{-1})^*
\end{equation}
where $\iota^*$ is the pull-back by $\iota\colon \gg^* \times \gg \to \gg \times \gg^*$, 
$ \iota(\xi, X)=(X, \xi)$.

We need a natural representation on $M$ by unitary operators  in $\Yc = L^2(\gg)$, given by
the natural induced representation described in subsection~\ref{induced-repres}.  
Namely,  $\pi\colon M \to \Bc (\Yc)$ is given by
\begin{equation}\label{repres_weyl}
\pi(\varphi, g) f(X)=\ee^{\ie \varphi(\exp_G X)} f ((-\log_G g)\ast X), \qquad f \in {\Yc}. 
\end{equation}
Then $\pi(\varphi, g) $ is unitary for every $(\varphi, g)\in M$.

Consider now $\theta_0\colon \Xi \to \Fc$ a Borel measurable function.
Then we set
\begin{equation}\label{theta_weyl}
\theta\colon \Xi \to \Lie(M), \qquad \theta(X, \xi) 
= Z_{\theta_0(X, \xi), X}, \quad (X, \xi) \in \Xi= \gg \times \gg^*,
\end{equation}
where  $Z_{\varphi, X}$ with $ (\varphi, X)\in \Fc \times \gg$ has been defined in Proposition~\ref{top3.6}.

We consider the Weyl calculus for $\pi$ along  $\theta$ above.
Recall that when $a \in F_\Xi L^1(\Xi)$
\begin{equation}\label{weyl_loc_special}
\Op^\theta(a)f
=\int\limits_{\Xi} \check{a}(X, \xi)\pi(\exp_M\theta(X, \xi))f\,\de(X, \xi) \qquad
f\in\Yc.
\end{equation}
We see that here
$$ \exp_M \theta(X, \xi) = \theta(X, \xi)(1) = \big ( \int\limits_{0}^1 
\lambda_{\exp_G(sX)} \theta_0(X, \xi)\, \de s, \, \exp_G X), $$
hence
\begin{equation}\label{weyl_pi}
\begin{aligned}
\pi(\exp_M \theta(X, \xi))f (Y)&  = \ee^{\ie   \int\limits_{0}^1  
\theta_0(X, \xi)(\exp_G (-s X) \exp_G Y )\, \de s} f((-X)\ast Y)\\
&=\ee^{\ie   \int\limits_{0}^1  \theta_0(X, \xi)(\exp_G ((-sX) \ast Y) )\, \de s} f((-X)\ast Y)\\
\end{aligned}
\end{equation}
when $f \in \Yc$.

We have thus obtained
\begin{equation}\label{weyl1001}
\Op^\theta(a) f (Y) = 
\int\limits_\Xi \check a(X, \xi) \ee^{\ie   \int\limits_{0}^1  
\theta_0(X, \xi)(\exp_G ( (-sX) \ast Y) )\, \de s} f( (-X)\ast Y)
\, \de(X, \xi).
\end{equation} 
By changing variables 
we get that
\begin{equation}\label{weyl1002}
\Op^\theta(a) f (Y) = \int\limits_\Xi \check a( Y \ast (-Z) , \xi) 
\ee^{\ie   \int\limits_{0}^1  \theta_0( Y\ast (- Z), \xi)(\exp_G ( (s (Z\ast( -Y))) \ast Y ) )\, \de s} 
f(Z)
\, \de(Z, \xi).
\end{equation}
We may use Fubini's theorem to see that  
the operator 
$\Op^\theta(a)$ is an integral operator with kernel
\begin{equation}\label{K_a:1}
 K_a(Y, Z) = \int\limits_{\gg^*} \check a(Y\ast (- Z), \xi) 
\ee^{\ie   \int\limits_{0}^1  \theta_0(Y\ast (-Z), \xi)
(\exp_G ((s (Z\ast (-Y)))\ast Y ) )\, \de s} \, \de \xi. 
\end{equation}

In  the case where $\theta_0$ is of the form
\begin{equation} \label{theta-magn}
 \theta_0(X, \xi) (x) = \scalar{\xi}{\log_G x}+\scalar{A(\log_G x)}{(R_{\log_G x})'_0X}
 \end{equation}
where $A\colon \gg \to \gg^*$ is continuous and $x\mapsto \scalar{A(\log_G x)}{(R_{\log_G x})'_0X}$
belongs to $\Fc$ for every $X\in \gg$, the expressions above can be further simplified. 
Denote 
\begin{equation}\label{alpha_A}
\alpha_A(Y, Z) = \exp\Bigl({\ie \int\limits_0^1 \scalar{A((s(Z\ast(-Y)))\ast Y)}{(R_{(s(Z\ast(-Y)))\ast Y})'_0(Y\ast (-Z))} \, \de s}\Bigr).
\end{equation}
This is a continuous  complex valued function on $\gg \times \gg$.
With this notation \eqref{K_a:1} becomes
$$ K_a(Y, Z) =   \alpha_A(Y, Z) \int\limits_{\gg^*} \check a(Y\ast (-Z), \xi) 
\ee^{\ie  \int \limits_0^1 \scalar{\xi}{(s(Z\ast (-Y))) \ast Y}  \, \de s}\de \xi . $$
Hence, in the case where $\theta_0$ is as in \eqref{theta-magn}, we get
\begin{equation}\label{K_a:2}
\begin{aligned}
K_a(Y, Z) & = \alpha_A(Y, Z) \int \limits_{\gg^*} 
(F_\gg \otimes F_{\gg ^*}  a (\xi, Y\ast (-Z)) \ee^{\ie \int \limits_0^1 \scalar{\xi}{(s(Z\ast (-Y))) \ast Y}  \, \de s} \, \de \xi\\
& = \alpha_A(Y, Z)(1\otimes F_{\gg}^{-1})a ( \int \limits_0^1 (s(Z\ast (-Y))) \ast Y) \, \de s, Y\ast (-Z)).
\end{aligned}
\end{equation}

\begin{definition}\label{magn_calc}
\normalfont
In the setting of subsection~\ref{magn_quant}, 
the simply connected nilpotent Lie group $G$ is identified with its Lie algebra $\gg$ 
by means of the exponential map. 
Let $\Fc$ an admissible space of functions on $\gg$ 
which contains both $\gg^*$ and the constant functions. 
Assume that $A\in\Omega^1(\gg)$ is 
a magnetic potential such that 
for every $X\in\gg$ the function $Y\mapsto\langle A_Y,(R_Y)'_0X\rangle$ belongs to $\Fc$ and define 
$\theta_0\colon\gg\times\gg^*\to\Fc$ as in \eqref{theta-magn} 
(or Proposition~\ref{para}).  
Then for every $a\in\Sc(\gg\times\gg^*)$  
there exists a linear operator $\Op^\theta(a)$ 
in $L^2(\gg)$ defined by \eqref{weyl1002}. 
We will denote $\Op^A(a):=\Op^\theta(a)$ and will call it 
a \emph{magnetic pseudo-differential operator} with respect to 
the magnetic potential~$A$. 
The function $a$ is the \emph{magnetic Weyl symbol} of the pseudo-differential operator $\Op^A(a)$, 
and the \emph{Weyl calculus with respect to the magnetic potential $A$} is 
the mapping $\Op^A$ which takes a function $a\in\Sc(\gg\times\gg^*)$ 
into the corresponding pseudo-differential operator. 
\qed
\end{definition}

\begin{theorem}\label{main1}
Assume that $A\in\Omega^1(\gg)$ is 
a magnetic potential such that 
for every $X\in\gg$ the function $Y\mapsto\langle A_Y,(R_Y)'_0X\rangle$ belongs to $\Fc$. 
The the Weyl calculus $\Op^A$ has the following properties: 
\begin{enumerate}
\item\label{main1_item0} 
For $P_0\in\gg$ let $A(Q)P_0$ be  the multiplication operator defined by 
the function $Y\mapsto\scalar{A_Y}{(R_Y)'_0P_0}$. 
Then the usual functional calculus for the self-adjoint operator 
$-\ie\dot{\lambda}(P_0)+ A(Q)P_0$ in $L^2(\gg)$  
can be recovered from $\Op^A$.
\item\label{main1_item1_1/2} 
Gauge covariance with respect to the magnetic potential $A$: 
If $A_1\in\Omega^1(\gg)$ is another magnetic potential with
$dA=dA_1\in\Omega^2(\gg)$ 
and the function $Y\mapsto\langle A_Y,(R_Y)'_0X\rangle$ belongs to $\Fc$ for every $X\in\gg$, 
then there exists 
$\psi\in\Fc$ such that unitary operator $U\colon L^2(\gg)\to L^2(\gg)$ 
defined by the multiplication by $\ee^{\ie\psi}$ satisfies 
$U\Op^A(a)U^{-1}=\Op^{A_1}(a)$ for every symbol $a\in\Sc(\gg\times\gg^*)$. 
 \item\label{main1_item1}
 If $\Cpol(\gg)\subseteq\Fc$ and the function $Y\mapsto\langle A_Y,(R_Y)'_0X\rangle$ belongs to $\Cpol(\gg)$ 
for every $X\in\gg$, 
then
for every $a\in\Sc(\gg\times\gg^*)$ the magnetic pseudo-differential operator $\Op^A(a)$ 
is  bounded linear on $L^2(\gg)$ and is defined by an integral kernel 
$K_a\in\Sc(\gg\times\gg)$ given by formula \eqref{K_a:2}. 
\item\label{main1_item2} 
Under the hypothesis of \eqref{main1_item1} 
the correspondence  $a\mapsto K_a$ is an isomorphism of Fr\'echet spaces 
$\Sc(\gg\times\gg^*)\to\Sc(\gg\times\gg)$
and extends to a unitary operator 
$L^2(\gg\times\gg^*)\to L^2(\gg\times\gg)$. 
\end{enumerate}
\end{theorem}

\begin{proof}
Assertion~\eqref{main1_item0} follows by Remark~\ref{top5}\eqref{top5_item2} along with the fact that 
for the representation~\eqref{repres_weyl} 
we get by \eqref{weyl_pi}
$$\begin{aligned}
\pi(\exp_M(\theta(tP_0,0)))f(Y)
&=\ee^{\ie\int\limits_{0}^1  \theta_0(tP_0,0)((-stP_0) \ast Y)\, \de s} f((-tP_0)\ast Y)
\end{aligned}
$$
where 
$$\begin{aligned}
\int\limits_0^1\theta_0(tP_0,0)((-stP_0) \ast Y)\, \de s 
&=\int\limits_0^1\scalar{A((-stP_0) \ast Y)}{(R_{(-stP_0) \ast Y})'_0(tP_0)}\,
\de s \\
&=\int\limits_0^t\scalar{A((-sP_0) \ast Y)}{(R_{(-sP_0)\ast Y})'_0P_0}\,\de s. 
\end{aligned}$$
Hence 
$$\frac{\de}{\de t}\Big\vert_{t=0}\pi(\exp_M(\theta(tP_0,0)))f
=\dot\lambda(P_0)f+\ie (A(Q)P_0)f $$
for $f\in L^2(\gg)$ such that the right-hand side belongs to~$L^2(\gg)$. 
See Remark~\ref{top5}\eqref{top5_item2} for the way the functional calculus 
of the self-adjoint operator $-\ie\dot\lambda(P_0)+A(Q)P_0$ can be recovered. 

For assertion~\eqref{main1_item1_1/2} note that if $d(A-A_1)=0$ on $\gg$ 
hence we have $d\psi=A-A_1$ for the function $\psi\colon\gg\to\RR$ defined 
by $\psi(X)=\int\limits_0^1\langle (A-A_1)_{tX},X\rangle \, \de t$. 
In particular $\psi\in\Fc$ and it follows by Proposition~\ref{top3.6}\eqref{top3.6_item2.5} 
that in the group $M=\Fc\rtimes_\lambda\gg$ we have 
$$(\Ad_M\psi)\theta(X,\xi)=\theta_1(X,\xi)$$ 
for every $X\in\gg$ and $\xi\in\gg^*$, 
where $\theta_1(X,\xi)$ is obtained as in \eqref{theta_weyl} 
with $\theta_0(X,\xi)$ replaced by the function 
$Y\mapsto\xi(Y)+\langle (A_1)_Y,(R_Y)'_0X\rangle$. 
Now Remark~\ref{top5}\eqref{top5_item3} shows that the assertion holds 
with $U=\pi(\psi)\colon L^2(\gg)\to L^2(\gg)$. 
Also note that, according to~\eqref{repres_weyl}, 
$U$ is actually the multiplication operator by the function $\ee^{\ie\psi}$.

Now assume the hypothesis of assertions \eqref{main1_item1}~and~\eqref{main1_item2} 
and remember that 
the first of these properties had been already proved 
in the discussion preceding Definition~\ref{magn_calc}. 
Further note that  
$\alpha_A(\cdot), \alpha_A(\cdot)^{-1}\in\Cpol(\gg\times\gg)$ by~\eqref{alpha_A}. 
Since moreover $\vert\alpha(\cdot)\vert=1$, we see from formula~\eqref{K_a:2} 
that in order to show prove the asserted properties of the correspondence $a\mapsto K_a$ 
it will be enough to check that the mapping 
\begin{equation*}
\Sigma\colon\gg\times\gg\to\gg\times\gg,\quad 
\Sigma(Y,Z)=\Bigl( \int \limits_0^1 (s(Z\ast (-Y))) \ast Y) \, \de s, Y\ast (-Z) \Bigr)
\end{equation*}
is a polynomial diffeomorphism whose inverse is polynomial and 
which preserves the Lebesgue measure on $\gg\times\gg$. 
For this purpose let us note that $\Sigma=\Sigma_2\circ\Sigma_1$, where the mappings 
$\Sigma_1,\Sigma_2\colon\gg\times\gg\to\gg\times\gg$ are defined by 
$$
\Sigma_1(Y,Z)=(-Y,Y\ast(-Z))
\quad\text{and}\quad 
\Sigma_2(V,W)=\Bigl(-\int\limits_0^1 V\ast(sW)\,\de s,W\Bigr).
$$ 
By using the fact that $\gg$ is a nilpotent Lie algebra it is straightforward 
to prove that $\Sigma_1$ is a measure-preserving polynomial diffeomorphism whose inverse is polynomial, 
and so is $\Sigma_2$ because of Proposition~\ref{nilp2}. 
This completes the proof. 
\end{proof}

\begin{definition}\label{diez}
\normalfont
If we assume that $\Cpol(\gg)\subseteq\Fc$ and $A\in\Omega^1(\gg)$ 
has the property that 
the function $Y\mapsto\langle A_Y,(R_Y)'_0X\rangle$ belongs to $\Cpol(\gg)$ 
for every $X\in\gg$, 
then it follows by Theorem~\ref{main1}\eqref{main1_item1} that 
for every $a_1,a_2\in\Sc(\gg\times\gg^*)$ there exists a unique function 
$a_1\#^A  a_2\in\Sc(\gg\times\gg^*)$ such that 
$\Op^A(a_1)\Op^A(a_2)=\Op^A(a_1\#^A  a_2)$ 
and the \emph{magnetic Moyal product}
$$\Sc(\gg\times\gg^*)\times\Sc(\gg\times\gg^*)\to \Sc(\gg\times\gg^*),\quad 
(a_1,a_2)\mapsto a_1\#^A  a_2 $$
is a bilinear continuous mapping. 
For the sake of simplicity we denote $a_1\#a_2:=a_1\#^A  a_2$ 
whenever the magnetic potential $A$ had been already specified.
\qed
\end{definition}

\subsection{The magnetic Moyal product for two-step nilpotent Lie algebras}\label{2step}

In this subsection we shall assume that $\gg$ is a \emph{two-step nilpotent Lie algebra}, 
that is, $[\gg,[\gg,\gg]]=\{0\}$ and moreover $\Cpol(\gg)\subseteq\Fc$ and $A\in\Omega^1(\gg)$ 
is a magnetic potential such that $\langle A(\cdot),X\rangle\in\Cpol(\gg)$ for every $X\in\gg$. 

\begin{lemma}\label{2step_lemma}
The following assertions hold in the two-step nilpotent Lie algebra $\gg$: 
\begin{enumerate}
\item\label{2step_lemma_item1} 
For every $X,Y\in\gg$ we have 
$(s(X\ast (-Y))) \ast Y=sX+(1-s)Y$ for arbitrary $s\in\RR$ and 
$\int \limits_0^1 (s(X\ast (-Y))) \ast Y \, \de s=\frac{1}{2}(X+Y)$. 
\item\label{2step_lemma_item2} 
For arbitrary $X,Y,Z,T\in\gg$ we have 
$$\begin{cases}
X=\frac{1}{2}(Y+Z) & \\
T=Y\ast(-Z)
\end{cases}
\iff 
\begin{cases}
Y=(\frac{1}{2}T) \ast X &\\
Z=(-\frac{1}{2}T)\ast X.
\end{cases} $$
Moreover, the diffeomorphism $\gg\times\gg\to\gg\times\gg$, $(Y,Z)\mapsto(\frac{1}{2}(Y+Z),Y\ast(-Z))$ 
preserves the Lebesgue measure. 
\item\label{2step_lemma_item3} 
For arbitrary $X,Z,T,z,t\in\gg$ we have 
$$\begin{cases}
\frac{1}{2}((\frac{1}{2}T)\ast X)+Z)=z & \\
\frac{1}{2}(Z+((-\frac{1}{2}T)\ast X))=t
\end{cases}
\iff 
\begin{cases}
T=2(z-t) + [X, z-t] &\\
Z=z+t-X.
\end{cases} $$
Moreover, the diffeomorphism $\gg\times\gg\to\gg\times\gg$, $(Z, T)\mapsto(z, t)$ 
preserves the Lebesgue measure. 
\end{enumerate}
\end{lemma}

\begin{proof}
\eqref{2step_lemma_item1} 
Indeed, for arbitrary $s\in\RR$ we have 
$$\begin{aligned}
(s(X\ast (-Y))) \ast Y
&= (s(X-Y-\frac{1}{2}[X,Y]))) \ast Y  
 = sX-sY-\frac{s}{2}[X,Y]+Y+ \frac{s}{2}[X,Y]\\
&=sX+(1-s)Y.
\end{aligned}$$
\eqref{2step_lemma_item2} 
For the implication ``$\Rightarrow$'' note that $T=Y\ast(-Z)$ 
actually means $T=Y-Z-\frac{1}{2}[Y,Z]$. 
If we apply $-\ad_{\gg}Z$ to both sides of the latter equation, 
then we get $[T,Z]=[Y,Z]$, and then 
\begin{equation}\label{2step_lemma_eq}
T=Y-Z-\frac{1}{2}[T,Z].
\end{equation}
On the other hand, the first of the assumed equations implies $2X=Y+Z$, 
and then we can eliminate $Y$ between this equation and~\eqref{2step_lemma_eq}. 
We thus get $2X-T=2Z+\frac{1}{2}[T,Z]$, 
and then by applying $\ad_{\gg}T$ to both sides of this equality we get $[T,X]=[T,Z]$. 
It then follows by~\eqref{2step_lemma_eq} that $T=Y-Z-\frac{1}{2}[T,X]$. 
Since $2X=Y+Z$, we get 
$$\begin{cases}
Y=X+\frac{1}{2}T+\frac{1}{4}[T,X]=(\frac{1}{2}T)\ast X & \\
Z=X-\frac{1}{2}T-\frac{1}{4}[T,X]=(-\frac{1}{2}T)\ast X.
\end{cases}$$
This concludes the proof of the implication ``$\Rightarrow$'', 
and the converse implication can be easily proved in a similar manner. 
The assertion regarding the measure-preserving property can be 
easily checked by computing the Jacobian of the diffeomorphism. 

\eqref{2step_lemma_item3}
We have 
$$\begin{cases}
\frac{1}{2}(((\frac{1}{2}T)\ast X)+Z)=z & \\
\frac{1}{2}(Z+((-\frac{1}{2}T)\ast X))=t
\end{cases}
\iff 
\begin{cases}
X+\frac{1}{2}T-\frac{1}{4}[X,T]+Z=2z &\\
X-\frac{1}{2}T+\frac{1}{4}[X,T]+Z=2t,
\end{cases} $$
and now the conclusion follows at once.
\end{proof}

\begin{theorem}\label{2step_th}
Assume that $\gg$ is a two-step nilpotent Lie algebra, 
$\Cpol(\gg)\subseteq\Fc$, and $A\in\Omega^1(\gg)$ 
is a magnetic potential such that $\langle A(\cdot),X\rangle\in\Cpol(\gg)$ for every $X\in\gg$. 
Then the following assertions hold: 
\begin{enumerate}
\item\label{2step_th_item1} 
For every $a\in\Sc(\gg\times\gg^*)$ 
the integral kernel of the operator $\Op^A(a)\colon L^2(\gg)\to L^2(\gg)$ 
is given by the formula 
\begin{equation}\label{K_a:2step}
K_a(Y,Z)=\alpha_A(Y,Z)(1\otimes F_{\gg}^{-1})a (\frac{1}{2}(Y+Z), Y\ast (-Z))
\end{equation}
where 
\begin{equation}\label{alpha_A:2step}
\alpha_A(Y, Z) = \exp\Bigl({-\ie \int\limits_0^1\langle A(sZ+(1-s)Y),Z\ast(-Y)\rangle \, \de s}\Bigr)
\end{equation}
for every $Y,Z\in\gg$.
\item\label{2step_th_item2} 
Set 
$$\beta_A\colon\gg\times\gg\times\gg\to\CC,\quad 
\beta(X,Y,Z)=\alpha_A^{-1}(X,Y)\alpha_A(Y,Z)\alpha_A(Z,X). $$
If $a,b\in\Sc(\gg\times\gg^*)$ then 
\begin{equation}\label{diez:2step}
\begin{aligned}
(a\#b)(X,\xi)
=\iiiint\limits_{\gg\times\gg\times\gg^*\times\gg^*}
&a(Z,\zeta)b(T,\tau)\ee^{2 \ie \scalar{(Z-X,\zeta-\xi)}{(T-X,\tau-\xi)}} 
\ee^{-\ie (\scalar{\xi+\zeta}{[X,Z]}+ \scalar{\zeta+\tau}{[Z,T]}+
\scalar{\tau+\xi}{[T,X]})}
\times 
\\
&\beta_A(Z-T+X, T-Z+X,Z+T-X)
\,\de\zeta\de\tau\de Z\de T
\end{aligned}
\end{equation}
for every $(X,\xi)\in\gg\times\gg^*$. 
\end{enumerate}
\end{theorem}

\begin{proof}
Formulas \eqref{K_a:2step}~and~\eqref{alpha_A:2step} follow at once 
by using \eqref{K_a:2}~and~\eqref{alpha_A}, respectively, 
and taking into account Lemma~\ref{2step_lemma}\eqref{2step_lemma_item1}. 

In order to prove~\eqref{diez:2step}, note first that by \eqref{K_a:2step} 
and Lemma~\ref{2step_lemma}\eqref{2step_lemma_item2} 
we have for every $c\in\Sc(\gg\times\gg^*)$ 
and $X,T\in\gg$ the equation 
$K_c((\frac{1}{2}T)\ast X,(-\frac{1}{2}T)\ast X)
=\alpha_A((\frac{1}{2}T)\ast X,(-\frac{1}{2}T)\ast X)(1\otimes F_{\gg}^{-1})c(X,T)$, 
whence 
\begin{equation*}
c(X,\xi)=\int\limits_{\gg}\ee^{-\ie\langle\xi,T\rangle}
(\alpha_A^{-1}K_c)((\frac{1}{2}T)\ast X ,(-\frac{1}{2}T)\ast X)\,\de T.
\end{equation*}
for every $c\in\Sc(\gg\times\gg^*)$. 
Hence, by using the well-known formula for the integral kernel of the product 
of two operators defined by integral kernels, we get 
\begin{equation*}
\begin{aligned}
(a\#b)(X,\xi)
&=\int\limits_{\gg}\ee^{-\ie\langle\xi,T\rangle}
(\alpha_A^{-1}K_{a\#b})((\frac{1}{2}T)\ast X ,(-\frac{1}{2}T)\ast X)\,\de T \\
&=\int\limits_{\gg}\int\limits_{\gg}\ee^{-\ie\langle\xi,T\rangle}
\alpha_A^{-1}((\frac{1}{2}T)\ast X ,(-\frac{1}{2}T)\ast X)
K_a((\frac{1}{2}T)\ast X,Z)K_b(Z,(-\frac{1}{2}T)\ast X)\,\de Z\de T 
\end{aligned}
\end{equation*}
On the other hand, by~\eqref{K_a:2step} we get 
$$\begin{aligned}
   K_a((\frac{1}{2}T)\ast X,Z)
&=\alpha_A((\frac{1}{2}T)\ast X,Z)
   (1\otimes F_{\gg}^{-1})
   a(\frac{1}{2}(((\frac{1}{2}T)\ast X)+Z),(\frac{1}{2}T)\ast X\ast (-Z))  \\
&=\alpha_A((\frac{1}{2}T)\ast X,Z)
\int\limits_{\gg^*}\ee^{\ie\scalar{\zeta}{(\frac{1}{2}T)\ast X\ast (-Z)}}
   a(\frac{1}{2}(((\frac{1}{2}T)\ast X )+Z),\zeta) \,\de \zeta
\end{aligned}$$
and also by~\eqref{K_a:2step} we have similarly 
$$\begin{aligned}
K_b(Z,(-\frac{1}{2}T)\ast X)
&=\alpha_A(Z,(-\frac{1}{2}T)\ast X)
  (1\otimes F_{\gg}^{-1})
  b(\frac{1}{2}(Z+((-\frac{1}{2}T)\ast X)),Z\ast(- X)\ast (\frac{1}{2}T)) \\
&=\alpha_A(Z,(-\frac{1}{2}T)\ast X)\int\limits_{\gg^*}
\ee^{\ie\scalar{\tau}{Z\ast (-X)\ast (\frac{1}{2}T)}}
   b(\frac{1}{2}(Z+((-\frac{1}{2}T)\ast X)),\tau)\,\de\tau
\end{aligned}$$
We plug in these formulas in the above expression of the magnetic Moyal product $a\#b$ and get 
\begin{equation*}
\begin{aligned}
(a\#b)(X,\xi)
=\int\limits_{\gg}&\int\limits_{\gg}\int\limits_{\gg^*}\int\limits_{\gg^*}
\alpha_A^{-1}((\frac{1}{2}T)\ast X,(-\frac{1}{2}T)\ast X)
\alpha_A((\frac{1}{2}T)\ast X,Z)\alpha_A(Z,(-\frac{1}{2}T)\ast X)\times \\
&\ee^{\ie E(\zeta,\tau,Z,T)}
a(\frac{1}{2}(((\frac{1}{2}T)\ast X)+Z),\zeta)
b(\frac{1}{2}(Z+((-\frac{1}{2}T)\ast X)),\tau)\,\de\zeta\de\tau\de Z\de T
\end{aligned}
\end{equation*}
where 
\begin{equation*}
E(\zeta,\tau,Z,T)=-\scalar{\xi}{T}+\scalar{\zeta}{(\frac{1}{2}T)\ast X \ast(-Z)}
+\scalar{\tau}{Z\ast (-X)\ast (\frac{1}{2}T)}.
\end{equation*}
We change of variables $(Z,T)\mapsto(z,t)$ of Lemma~\ref{2step_lemma}\eqref{2step_lemma_item3}.
In these new variables we have
$$ \begin{aligned} 
(\frac{1}{2}T) \ast X  &= 2z -Z = z-t+X,\\
(-\frac{1}{2}T) \ast X  &= 2t -Z = t-z+X.
\end{aligned}
$$
It follows that 
\begin{equation*}
\begin{aligned}
(a\#b)(X,\xi)
=\int\limits_{\gg}&\int\limits_{\gg}\int\limits_{\gg^*}\int\limits_{\gg^*}
\beta_A(z-t+X, t-z+X,z+t-X)\times \\
&\ee^{\ie E(\zeta,\tau,z+t-x, 2(z-t)+[X, z-t])}
a(z,\zeta)b(t,\tau)\,\de\zeta\de\tau\de z\de t.
\end{aligned}
\end{equation*}
Note that in the change of variables above we have
$$ 
\begin{aligned} (\frac{1}{2} T) \ast X\ast (-Z) & = (2z-Z) \ast(-Z)\\
& = 2(z-Z) +\frac{1}{2}[2z-Z, -Z]\\
& = 2 (X-t) + [z, X-t],
\end{aligned}
$$
and similarly 
$$ 
\begin{aligned} 
Z \ast (-X) \ast (\frac{1}{2} T)  & = Z\ast (Z-2t) \\
& = 2(Z-t) +\frac{1}{2}[Z, Z-2t]\\
& = 2 (z-X) + [t, z-X],
\end{aligned}
$$
Thus
$$\begin{aligned}
E(\zeta,\tau,z+t-x, 2(z-t)+[X, z-t] )
=
&-\scalar{\xi}{2(z-t)+[X, z-t]}\\ &  
+\scalar{\tau}{2(X-t) + [z, X-t]}
+\scalar{\tau}{2(z-X)+[t, z-X]}\\
=&\scalar{2(\tau-\xi)}{z-X}-\scalar{2(\zeta-\xi)}{t-X} \\
& - \ie (\scalar{\xi+\zeta}{[X,z]}+ \scalar{\zeta+\tau}{[z,t]}+
\scalar{\tau+\xi}{[t,X]}
\end{aligned}
$$
and this completes the proof of~\eqref{diez:2step}. 
\end{proof}

It is clear that in the case when $\gg$ is an abelian Lie algebra, 
formula~\eqref{diez:2step} specializes to the formula for 
the magnetic Moyal product on $\RR^n$; see \cite{KO04} and \cite{MP04}. 
If moreover the magnetic potential $A\in\Omega^1(\gg)$ vanishes, 
then one recovers the formula for the composition 
of pseudo-differential operators in the framework of the Weyl calculus; 
see Section~18.5 in \cite{Hor07}. 

\subsection*{Acknowledgment}
We are grateful to Anders~Melin, Benjamin~Cahen, and Horia~Cornean 
for their kind help.

\end{document}